\documentclass{elsarticle}

\usepackage{graphicx}
\usepackage{subfigure}
\usepackage[english,francais]{babel}

\usepackage{bm}
\usepackage{stmaryrd}
\usepackage{amssymb}
\usepackage{amsthm}
\usepackage{amsmath}
\def\R{\mathbb{R}}
\def\C{\mathbb{C}}
\def\dfrac{\displaystyle\frac}
\def\ep{\varepsilon}
\def\kpara{\boldsymbol{k}}
\def\hom{\hat\omega}

\newtheorem{proposition}{Proposition}[section]

\begin{document}

\title{Artificial magnetism and magnetoelectric coupling from dielectric layers}

\author{Yan Liu, S\'ebastien Guenneau and Boris Gralak}
\address{Institut Fresnel, CNRS, Aix-Marseille Universit\'e, Ecole centrale Marseille, Campus de
Saint-J\'er\^ome, 13013 Marseille, France}

\label{firstpage}

\begin{abstract}%{homogenization, transfer matrix, Sophus Lie theorem, magneto-coupling}
We investigate a high-order homogenization (HOH) algorithm for periodic multilayered stacks. The mathematical tool of choice is a transfer matrix method. Expressions for effective permeability,
permittivity and magnetoelectric coupling are explored by frequency power expansions.
On the physical side, this high-order homogenization uncovers a magnetoelectric coupling effect (odd order approximation) and artificial magnetism (even order approximation) in moderate contrast photonic crystals. Comparing the effective parameters' expressions of a stack with three layers against that of a stack with two layers, we note that the magnetoelectric coupling effect vanishes while the artificial magnetism can still be achieved in a center symmetric periodic structure. Furthermore, we numerically check the effective parameters through the dispersion law and transmission curves of a stack with two dielectric layers against that of an effective bianisotropic medium: they present a good agreement in the low frequency (acoustic) band until the first stop band, where the analyticity of the logarithm function of transfer matrix ($\log\{T\}$) breaks down.
\end{abstract}

\maketitle

%------------------------------------------------------------------------------------------------------------------------------------
\section{Introduction}

There is a vast amount of literature in the homogenization of periodic structures with
the classical effect of artificial anisotropy. A less well-known effect is artificial
magnetism and low frequency stop bands through averaging processes in high-contrast periodic structures
\cite{Brien02,Chered06,Felbacq04,Zhikov00}.

In layman terms, O'Brien and Pendry observed in 2002 that rather than using the LC-resonance of conducting split ring resonators to achieve a negative permeability as an average quantity \cite{Pendry99}, one can design periodic structures displaying some Mie resonance,
e.g. by considering a cubic array of dielectric spheres of high-refractive index \cite{Brien02}. These resonances give rise to the
heavy-photon bands in photonic crystals which are responsible for a low frequency stop band corresponding to a range of frequencies wherein the
effective permeability takes extreme values \cite{Felbacq04}. So-called high-contrast homogenization \cite{Zhikov00} predicts this effect which is associated with the lack of a lower bound for a frequency dependent effective parameter deduced from a spectral problem reminiscent of Helmholtz resonators in mechanics.

In this paper, we would like to achieve such a magnetic activity without a high-contrast material. The route we propose is based upon a homogenization approach for high-frequencies i.e. when the period of a multilayered structure approaches the wavelength of optical wave. The extension of classical homogenization theory \cite{Bensoussan78,Zhikov94,Milton02} to high frequencies is of pressing importance for physicists working in the emerging field of metamaterials, but applied mathematicians also show a keen interest in this topic \cite{Chered06,Craster10,Felbacq05}, where the periodic structures at sub-wavelength scales ($\lambda/10$ to $\lambda/6$) \cite{Brien02,Pendry99} can clearly be regarded as almost homogeneous.

The tool of choice for our one-dimensional model is the transfer matrix method, which allows for analytical formulae as shown in Section \ref{sec:math}, a high-order homogenization (HOH) method is proposed wherein Baker-Campbell-Hausdorff formula (BCH, an extension of Sophus Lie theorem) was implemented; we stress that ideas contained therein can be extended to two and three-dimensional periodic structures, at least for simple geometries, such as woodpile structures, of particular interest in photonics, see \cite{Gralak03}. Importantly, we not only achieve magnetic activity in moderate contrast dielectric structures, but also unveil some artificial bianisotropy in Section \ref{sec:hoh}, where a multilayered stack with an alternation of two dielectric layers is considered. As proposed by Pendry, the bianisotropy is yet another route towards negative refraction \cite{Pendry04}, however, it is noted that the artificial bianisotropy vanishes in a periodic stack with center symmetry, where an extension of
the HOH method applied to a stack with $m(\geq 3)$ alternative layers is explored in Section \ref{sec:bch}. Furthermore, a correction factor is investigated in Section \ref{sec:corc} to estimate the asymptotic error, which is proportional to $1/n^p$ with $p$ the approximation order. We then numerically check the effective parameters through the dispersion law and transmission curves between the multilayers and the effective medium in Section \ref{sec:num}: A good agreement in the low frequency band up to the first stop band verifies the equivalence of the two structures. We finally take a frequency power expansion of the transfer matrix, which is analytic in the whole complex plane in Section \ref{sec:Tseries}, and draw some conclusions in Section \ref{sec:ccs}.
%--------------------------------------------------------------------------------------------------------------------------------
\section{Mathematical setup of the problem} \label{sec:math}

Figure \ref{figure1}(a) shows a schematic diagram of periodic multilayered stack with a unit cell made of two homogeneous dielectric layers $\mathcal{L}_1$ and $\mathcal{L}_2$ of respective thicknesses $h_1$ and $h_2$ ($h=h_1+h_2=\tilde{h}_1/n+\tilde{h}_2/n$), where $n$ is the number of the unit cells in the stack. The permittivities of the two layers are assumed as $\ep_1$, $\ep_2$, and the permeabilities are $\mu_1=\mu_2=\mu_0$. It should be noted that the whole thickness of the stack $\tilde{h}(=nh)$ is a constant and comparable with the wavelength of light as $\tilde{h}/\lambda \approx 1$. Starting from the case when $n=1$, the stack is a most simple structure consisting of only two dielectric layers, the effective medium theory \cite{Brekhov80,Born80,Yariv84} cannot be applied when $h/\lambda \approx 1$. However, if we increase $n$ to a large enough constant, then the system contains $n$ times smaller unit cells, the thickness of which will be much smaller than the wavelength of light (e.g. $h/\lambda \ll 1$). Hence a homogeneous medium with permittivity $\ep_{\rm eff}$, permeability $\mu_{\rm eff}$ and bianisotropy $K_{\rm eff}$ shown in figure \ref{figure1}(b) can be assumed to behave as an effective medium for such a multilayer. The approximation of the multilayer by the effective medium will be more accurate for large $n$, a fact which will be proved in the following sections.
\begin{figure}[!htb]
    \centering
    \includegraphics[width=0.9\textwidth]{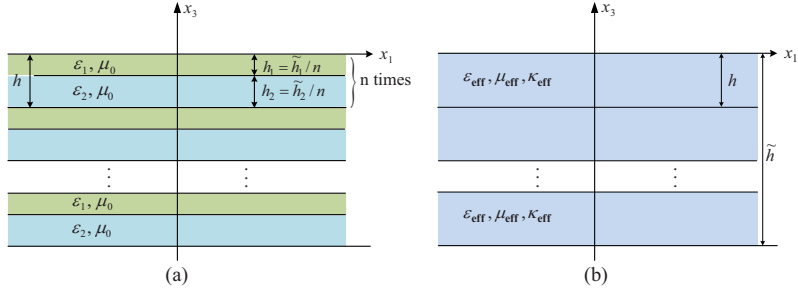}
    \caption{(a) Schematic diagram of a multilayered stack consisting of an alternation of two homogeneous dielectric layers of permittivities $\varepsilon_1$, $\varepsilon_2$ and thicknesses $h_1=\tilde{h}_1/n$, $h_2=\tilde{h}_1/n$, with $n$ the total number of the unit cells. The whole thickness of the stack is denoted by $\tilde{h}=nh=n(h_1+h_2)$, which is comparable with wavelength by $\tilde{h}/\lambda \approx 1$, when $n$ is large enough, then $\tilde{h}/(n\lambda) = h/\lambda \ll 1$, in other words, $\omega h/(2\pi c_0)\ll 1$ with $c_0$ the velocity of light in vacuum. (b) An effective medium described by anisotropic tensors of permittivity, permeability and bianisotropy (i.e. a metamaterial with artificial bianisotropy and magnetism), the thickness is $\tilde{h}$.}
\label{figure1}
\end{figure}
%-----------------------------------------------------------------------------------------------------------------------------------
\subsection{Time-harmonic Maxwell's equations}

At the oscillating frequency $\omega$, the electric and magnetic fields $\bf{E}$ and $\bf{H}$ are related to the electric and magnetic inductions $\bf{D}$ and $\bf{B}$ through the time-harmonic Maxwell's equations,
\begin{equation}
\boldsymbol{\nabla} \times {\bf H} ({\bf x})= - i \omega {\bf D} ({\bf x}) \, , \quad
\boldsymbol{\nabla} \times {\bf E} ({\bf x})= i \omega {\bf B} ({\bf x}) \, ,
\label{Maxwell}
\end{equation}
and the constitutive relations for non-magnetic isotropic dielectric media,
\begin{equation}
{\bf D}({\bf x}) = \varepsilon_m {\bf E} ({\bf x})\, , \quad {\bf B}({\bf x}) = {\bf H}({\bf x})  \, , \quad {\bf x} \, \in \, \mathcal{L}_m \, ,
\label{M0}
\end{equation}
where $\ep_m$ the permittivity in the $m^{th}$ homogeneous layer located in the domain $\mathcal{L}_m$ of $\mathbb{R}^3$.

Then a Fourier decomposition is introduced for both electric and magnetic fields as
\begin{equation}
\label{Fourier}
\begin{array}{l}
\widehat{\bf E}(k_1,k_2,x_3) = \dfrac{1}{2 \pi} \displaystyle\int_{\mathbb{R}^2} {\bf E}(x_1,x_2,x_3) \exp\big[- i (k_1 x_1 + k_2 x_2)] \, dx_1 dx_2 \, , \\[4mm]
\widehat{\bf H}(k_1,k_2,x_3) = \dfrac{1}{2 \pi} \displaystyle\int_{\mathbb{R}^2} {\bf H}(x_1,x_2,x_3) \exp\big[- i (k_1 x_1 + k_2 x_2)] \, dx_1 dx_2 \, .
\end{array}
\end{equation}
with $k_1$, $k_2$ the projections of wave vector ${\bf k}$ on $x_1$, $x_2$ axes, respectively, where an oblique polarizable plane wave with wave vector ${\bf k}=k_1 x_1+ k_2 x_2 +k_3 x_3$ is considered ($x_3$ axis is perpendicular to the layers).

Applying the decomposition (\ref{Fourier}) to equations (\ref{Maxwell}) and (\ref{M0}), we derive an ordinary differential equation (involving $4\times 4$-matrices and a 4-components column vector) \cite{Lakhtakia92}
\begin{equation}
\dfrac{\partial {\widehat{\text{F}}}}{\partial x_3} \, (\omega,k_1,k_2,x_3)
= i {\text{M}}_m(\omega, k_1,k_2) \, {\widehat{\text{F}}} (\omega,k_1,k_2,x_3) \, .
\label{dFdx3}
\end{equation}
where $\widehat{\text{F}}$ is a column vector containing the tangential components of the Fourier-transformed electromagnetic field $(\widehat{\bf E},\widehat{\bf H})$, i.e. the components along $x_1$ and $x_2$ axes.

Here, in order to apply a much simpler notation for subsequent calculation, we would like to define a new set of coordinates denoted by ($x_\parallel$, $x_\perp$, $x_3$): where the component $x_\parallel$ is along the direction of wave vector ${\boldsymbol k}=(k_1,k_2)$, $x_\perp$ is along ${\boldsymbol k}'=(-k_2,k_1)$ which is perpendicular to $x_\parallel$. In other words, the new set of coordinates is a rotation of the previous coordinates around the $x_3$ axis. For every vector $\bf{x}$, the change of coordinates from ($x_1$, $x_2$, $x_3$) to ($x_\parallel$, $x_\perp$, $x_3$) can be expressed by
\begin{equation}
  \left[ \begin{array}{l}
  x_\parallel \\[2mm]
  x_\perp
  \end{array}\right]= \dfrac{1}{\sqrt{k_1^2+k_2^2}}\left[
  \begin{array}{cc}
  k_1 & k_2 \\[2mm]
  -k_2 & k_1
  \end{array}\right]
  \left[ \begin{array}{l}
  x_1\\[2mm]
  x_2
  \end{array}\right].
\label{newcd}
\end{equation}
Note that, thanks to the symmetry of the geometry, the parameters of the multilayer are invariant under this transformation \cite{ward96}.

Hence, (\ref{dFdx3}) can be recast in the new coordinate system as
\begin{equation}
\dfrac{\partial {\widehat{\text{F}}}}{\partial x_3} \, (\omega,\kpara,x_3)
= i M_m(\omega, \kpara) \,{\widehat{\text{F}}} (\omega,\kpara,x_3) \, ,
\label{newdFdx3}
\end{equation}
with the column vectors
\begin{equation}\label{F}
{\widehat{\text{F}}} = \left[ \begin{array}{c} {\widehat{\text{F}}}_\parallel \\[2mm] {\widehat{\text{F}}}_\perp \end{array} \right] \, , \quad
{\widehat{\text{F}}}_\parallel = \left[ \begin{array}{c} \widehat{\text{E}}_\parallel \\[2mm] \widehat{\text{H}}_\parallel \end{array} \right] \, , \quad
{\widehat{\text{F}}}_\perp = \left[ \begin{array}{c} \widehat{\text{E}}_\perp \\[2mm] \widehat{\text{H}}_\perp \end{array} \right]  \, ,
\end{equation}
where $\widehat{\text{E}}_\parallel$, $\widehat{{\text H}}_\parallel$ (respectively $\widehat{\text{E}}_\perp$, $\widehat{\text{H}}_\perp$) are the components of the electric and magnetic fields along the $x_\parallel$ axis (respectively the $x_\perp$ axis).

Correspondingly, the matrix $M_m$ is a 4 by 4 matrix, the components of which can be expressed as
\begin{equation}
M_m(\omega,{\boldsymbol k}) = \omega \left[ \begin{array}{cc} \vspace*{1mm}
0 & \sigma_m + ({\boldsymbol k}^2 / \omega^2) \sigma_m^{-1} \\
- \sigma_m & 0 \end{array} \right] , \quad \sigma_m = \left[ \begin{array}{cc} \vspace*{1mm}
0 & \mu_0\\
- \varepsilon_m & 0 \end{array} \right] \,
\label{Mm}
\end{equation}
where $\kpara=\omega {\hat {\it k}}$ is the projection of wave vector ${\bf k}$ on $x_\parallel$, and $\hat {\it k}=\sin\theta_i$, with $\theta_i$ the incident angle.

The matrix $M_m(\omega,{\boldsymbol k})$ is independent of $x_3$ in each homogeneous layer, the solution of the equation (\ref{newdFdx3}) in the layer $\mathcal{L}_m$ is simply
\begin{equation}
{\widehat{\text{F}}} (\omega,{\boldsymbol k},x_3+h_m) = \exp\big[ i M_m(\omega,{\boldsymbol k}) h_m \big] \cdot {\widehat{\text{F}}} (\omega,{\boldsymbol k},x_3) \,
\label{solution}
\end{equation}
The exponential above is well-defined as a power series of matrix $M_m(\omega,{\boldsymbol k})$, and defines the transfer matrix in the $m^{th}$ layer of thickness $h_m$. Since this power series has infinite radius of convergence, the transfer matrix
\begin{equation}
  T_m(\omega,{\boldsymbol k})=\exp[iM_m h_m]
\end{equation}
is analytic with respect to the three independent variables $\omega$, $k_1$ and $k_2$. For an arbitrary permittivity profile (with the classical assumption of upper and lower bounds of permittivity greater than $\varepsilon_0$ uniformly in position ${\bf x}$ and number of layers $n$), analyticity is proved using a Dyson expansion \cite{Reed75}.
%-----------------------------------------------------------------------------------------------------------------------------------
\subsection{Main homogenization result}
From physical considerations, for the effective medium in figure \ref{figure1}(b), we postulate that the homogenized constitutive equations emerging from the asymptotic limit $n\rightarrow +\infty$ in the sequence of equations (\ref{M0}) in the new coordinate system are
\begin{equation}
\begin{array}{ccc}
{\bf D}_\text{\rm eff}({\boldsymbol k},x_3)& = \varepsilon_\text{\rm eff}(\omega,{\boldsymbol k}) {\bf E}_\text{\rm eff}({\boldsymbol k},x_3)
+i K_\text{\rm eff}(\omega,{\boldsymbol k}) \, J  {\bf H}_\text{\rm eff}({\boldsymbol k},x_3) \\[2mm]
{\bf B}_\text{\rm eff}({\boldsymbol k},x_3) &=  \mu_\text{\rm eff}(\omega,{\boldsymbol k})  {\bf H}_\text{\rm eff}({\boldsymbol k},x_3)+  i J K_\text{\rm eff}(\omega,{\boldsymbol k})  {\bf E}_\text{\rm eff}({\boldsymbol k},x_3)
\end{array}
\label{coneff}
\end{equation}
where $\varepsilon_\text{\rm eff}$, $\mu_\text{\rm eff}$ are tensors of rank two which represent respectively the (anisotropic) effective
permittivity, permeability
\begin{equation}
\varepsilon_\text{\rm eff} =
\left[ \begin{array}{ccc} \vspace*{1mm}
\varepsilon_\parallel & 0 & 0 \\
0 & \varepsilon_\perp & 0 \\
0 & 0 & \varepsilon_3
\end{array}
\right] , \quad
\mu_\text{\rm eff} =
\left[ \begin{array}{ccc} \vspace*{1mm}
\mu_\parallel & 0 & 0 \\
0 & \mu_\perp & 0 \\
0 & 0 & \mu_3
\end{array}
\right],
\label{peff}
\end{equation}
matrix $J$ corresponds to the 90 degrees rotation around the $x_3$ axis, and $K_{\rm eff}$ is the bianisotropic parameter measuring the magnetoelectric coupling effect
\begin{equation}
J =\left[ \begin{array}{ccc} \vspace*{1mm}
0 & -1 & 0 \\
1 & 0 & 0 \\
0 & 0 & 1
\end{array}
\right] , \quad
K_\text{\rm eff} =
\left[ \begin{array}{ccc} \vspace*{1mm}
K_\parallel & 0 & 0 \\
0 & K_\perp & 0 \\
0 & 0 & 0
\end{array}
\right] .
\label{JK}
\end{equation}
Now, applying equation (\ref{Fourier}) to (\ref{Maxwell}) and (\ref{coneff}), we obtain
\begin{equation}
\dfrac{\partial {\widehat{\text{F}}}}{\partial x_3} \, (\omega,{\boldsymbol k},x_3)
= i M_{\rm eff}(\omega, {\boldsymbol k}) \,{\widehat{\text{F}}} (\omega,{\boldsymbol k},x_3) \, ,
\label{dFdxeff}
\end{equation}
with matrix
\begin{equation}
 M_\text{eff}(\omega, {\boldsymbol k}) =
 \omega \left[ \begin{array}{cc}
 -i \sigma'_K & \sigma_\perp + ({\boldsymbol k}^2/\omega^2) \sigma_3^{-1} \\
 -\sigma_\parallel & i \sigma_K
 \end{array} \right] \, ,
\label{Meff}
\end{equation}
where the $2 \times 2$ blocs are defined by
\begin{equation}
\sigma_\parallel = \left[ \begin{array}{cc} 0 &  \mu_\parallel \\
- \varepsilon_\parallel & 0 \end{array} \right] , \quad
\sigma_\perp = \left[ \begin{array}{cc} 0 & \mu_\perp \\
- \varepsilon_\perp & 0 \end{array} \right] , \quad
\sigma_3 = \left[ \begin{array}{cc} 0 &  \mu_3 \\
- \varepsilon_3 & 0 \end{array} \right]
\label{sigmaeff}
\end{equation}
and
\begin{equation}
\sigma_K = \left[ \begin{array}{cc}
K_\perp & 0 \\
0 &  -K_\parallel  \end{array} \right] , \quad
\sigma'_K = \left[ \begin{array}{cc}
-K_\parallel & 0 \\
0 &  K_\perp \end{array} \right] .
\label{Keff}
\end{equation}
These parameters are all unknowns at this stage which we would like to derive from a homogenization algorithm.
The transfer matrix is correspondingly,
\begin{equation}
  T_{\rm eff}(\omega, {\boldsymbol k})=\exp[iM_{\rm eff} h].
\label{Teff}
\end{equation}

\section{High-order homogenization (HOH) algorithm for multilayered stack} \label{sec:hoh}
Since we have derived the transfer matrix for the multilayered stack and postulated its structure for the effective medium in the previous section, it follows that the description of the homogenization procedure shown in figure \ref{figure1} can be expressed as
\begin{equation}
\exp[i M_2 h_2] \exp[i M_1 h_1]
= \exp[i M_\text{eff} h] \, .
\label{eqn}
\end{equation}
This means the two structures should present a same transmission property, where $M_{\rm eff}$ is the unknown to be calculated. Note that the left side of equation (\ref{eqn}) is a product of two exponential functions, which can be approximated by introducing the Baker-Campbell-Hausdorff (BCH) formula (an extension of Sophus Lie theorem, see \cite{Weiss62}). In mathematics, the BCH formula is concerned with
\begin{equation}
\exp[Z] = \exp \big[ A_1 \big] \exp \big[ A_2 \big]
\end{equation}
with $A_1$ and $A_2$ square matrices. An analytical expression for $Z$ is
\begin{align}
  Z&=\log(\exp[A_1]\exp[A_2]) \notag \\
  &= A_1 + A_2 + \dfrac{1}{2}\llbracket A_1, A_2 \rrbracket + \dfrac{1}{12}\llbracket  A_1 , \llbracket A_1, A_2 \rrbracket \rrbracket - \dfrac{1}{12} \llbracket  A_2 , \llbracket A_1, A_2 \rrbracket \rrbracket + \cdots  \,
\label{BCH}
\end{align}
where $\llbracket A_1, A_2 \rrbracket = A_1 A_2 - A_2 A_1$ is the commutator of $A_1$ and $A_2$, the product of which is noncommutative with $\llbracket A_2, A_1 \rrbracket=-\llbracket A_1, A_2 \rrbracket$.
Here, we would like to denote $A_1 + A_2$ in (\ref{BCH}) as the zeroth order approximation for Z, which corresponds to the classical homogenization; $\llbracket A_1, A_2 \rrbracket/2$ the first order, $\llbracket  A_1 , \llbracket A_1, A_2 \rrbracket \rrbracket / 12 - \llbracket  A_2 , \llbracket A_1, A_2 \rrbracket \rrbracket/12$ the second order approximation, and so on.

\noindent From (\ref{BCH}) and (\ref{eqn}), we have
\begin{align}
i M_\text{eff} h &= i (M_2 h_2 + M_1 h_1) + \dfrac{1}{2}\llbracket i M_2 h_2, i M_1 h_1\rrbracket \notag \\[-2mm]
& + \dfrac{1}{12} \llbracket i M_2 h_2 , \llbracket iM_2 h_2 , iM_1 h_1 \rrbracket \rrbracket - \dfrac{1}{12} \llbracket i M_1 h_1, \llbracket i M_2 h_2 , i M_1 h_1\rrbracket \rrbracket + \cdots
\label{approxMeff}
\end{align}
Furthermore, the expressions for the effective parameters in (\ref{peff}) and (\ref{JK}) can be derived by comparing the two matrices in the left and right hand sides of (\ref{approxMeff}).

First, we consider the zeroth order approximation in (\ref{approxMeff}), it yields $M_\text{eff} \approx M_1 f_1 + M_2 f_1$ with the filling fractions $f_1(=h_1/h)$ and $f_2(=h_2/h)$, respectively, and the effective parameters are
\begin{equation}
\begin{array}{c}
\varepsilon_\parallel = \varepsilon_\perp = \varepsilon_1 f_1 + \varepsilon_2 f_2 \, , \quad
\varepsilon_3^{-1} = \varepsilon_1^{-1} f_1 + \varepsilon_2^{-1} f_2 \,   \\[2mm]
\mu_\parallel = \mu_\perp = \mu_3 = \mu_0 \, ,\quad  K_\parallel = K_\perp = 0 \, .
\end{array}
\label{zeroorder}
\end{equation}
They are identical to the effective permittivities presented in \cite{Born80,Yariv84,Thornburg57,Raguin93,Guenneau00} by classical homogenization: the effective permittivity, permeability are equal to the average of two dielectric layers, while the bianisotropy is zero.

If we go further by taking the first order approximation, we obtain
\begin{equation}
M_\text{eff} \approx M_1 f_1 + M_2 f_2 + \dfrac{ih}{2} \llbracket M_2 f_2, M_1 f_1\rrbracket .
\end{equation}
Since both $M_1$ and $M_2$ are off-diagonal matrices, then their commutator leads to a diagonal matrix, the components of which correspond to those of $M_{\rm eff}$ in (\ref{Meff}), i.e.
\begin{equation}
\dfrac{ih}{2}\llbracket M_2 f_2, M_1 f_1 \rrbracket \, = \omega \left[ \begin{array} {cc}
- i\sigma'_K & 0 \\ 0 & i\sigma_K
\end{array} \right]
\label{commutator}
\end{equation}
where
\begin{equation}
\sigma_K = \omega h  \dfrac{\varepsilon_1 - \varepsilon_2}{2} f_1 f_2
\left[ \begin{array} {cc}
\mu_0& 0 \\ 0 & - \mu_0 + \dfrac{{\boldsymbol k}^2}{\omega^2}\dfrac{\varepsilon_1 + \varepsilon_2}{\varepsilon_1 \varepsilon_2} \end{array} \right]
\label{sigmaxi}
\end{equation}
provides the first order correction to the leading order approximation (classical homogenization) in (\ref{zeroorder}). This first order correction is encompassed in the following bianisotropic parameter
\begin{equation}
\begin{array}{c}
  K_\perp(\omega, \kpara)= \dfrac{\omega h}{2}\mu_0(\ep_1-\ep_2) f_1f_2 , \\
  K_\parallel(\omega,\kpara)=\dfrac{\omega h}{2}\mu_0(\ep_1-\ep_2) f_1f_2 \left[1-\dfrac{{\boldsymbol k}^2}{{ \omega}^2}\dfrac{\ep_1+\ep_2}{\mu_0\ep_1\ep_2}\right].
\end{array}
\label{K}
\end{equation}
\noindent Notice that $K$ is not only frequency dependent but also exhibits spatial dispersion. It leads to $K_\parallel \neq K_\perp$ when $\kpara\neq0$.

Furthermore, if we consider the second order correction, a term with "double commutator" will appear in the asymptotic expansion
\begin{align}
M_\text{eff} & \approx (M_2 f_2 + M_1 f_1) -\dfrac{ih}{2} \llbracket M_2 f_2, M_1 f_1 \rrbracket \notag \\[-2mm]
&+\dfrac{h^2}{12} \llbracket M_1 f_1 , \llbracket  M_2 f_2 , M_1 f_1 \rrbracket \rrbracket - \dfrac{h^2}{12} \llbracket M_2 f_2 , \llbracket  M_2 f_2 , M_1 f_1 \rrbracket \rrbracket,
\label{second}
\end{align}
with commutator of $M_1$ and $M_2$ given in (\ref{commutator}), and double commutator
\begin{equation}
\llbracket M_1f_1, \llbracket M_2 f_2, M_1 f_1 \rrbracket \rrbracket = \dfrac{2\omega^2}{h}f_1
\left[ \begin{array} {cc}
0 & \sigma_1 \sigma_K + \sigma'_K \sigma_1+\dfrac{\kpara^2}{\omega^2} \, (\sigma_1^{-1} \sigma_K + \sigma'_K \sigma_1^{-1})
 \\ \sigma_1 \sigma'_K + \sigma_K \sigma_1 & 0
\end{array} \right]
\label{dbcommutator}
\end{equation}
According to the definitions of $\sigma_m$ in (\ref{Mm}), $\sigma_K$ and $\sigma'_K$ in (\ref{Keff}), we have
\begin{equation}
\sigma_1 \sigma'_K = \sigma_K \sigma_1  \, , \quad \sigma_1 \sigma_K = \sigma'_K \sigma_1 \, , \quad
\sigma_1^{-1} \sigma_K = \sigma'_K \sigma_1^{-1}\,.
\end{equation}
Thus, (\ref{dbcommutator}) can be simplified to
\begin{equation}
\llbracket M_1 f_1, \llbracket M_2 f_2, M_1 f_1 \rrbracket \rrbracket = \dfrac{4\omega^2}{h}f_1 \left[ \begin{array} {cc}
0 & \sigma_1 \sigma_K + \dfrac{\kpara^2}{\omega^2}\sigma_1^{-1} \sigma_K \\ \sigma_K \sigma_1 & 0
\end{array} \right].
\label{doublecommutator2-1}
\end{equation}
Similar equalities hold for $ \llbracket M_2 f_2, \llbracket M_2 f_2, M_1 f_1 \rrbracket \rrbracket$. Hence, the terms arising from "double commutator" lead to
\begin{equation}
\begin{array}{l}
\dfrac{h^2}{12} \left( \llbracket M_1 f_1, \llbracket M_2 f_2, M_1 f_1 \rrbracket \rrbracket -  \llbracket M_2 f_2, \llbracket M_2 f_2, M_1 f_1 \rrbracket \rrbracket \right) \\[3mm]
= \dfrac{\omega^2 h}{3} \left[ \begin{array} {cc}
0 & \sigma_1 \sigma_K f_1 - \sigma_2 \sigma_K f_2 + \dfrac{{\boldsymbol k}^2}{\omega^2} (\sigma_1^{-1} \sigma_K f_1 - \sigma_2^{-1} \sigma_K f_2) \\ \sigma_K \sigma_1 f_1 -
\sigma_K \sigma_2 f_2 & 0
\end{array} \right]
\end{array}
\label{addterm}
\end{equation}
which is again an off-diagonal matrix. Substituting (\ref{addterm}) into (\ref{second}) and
comparing with the form of matrix $M_{\rm eff}$ in (\ref{Meff}), we find
\begin{equation}
\begin{array}{l}
-\sigma_\parallel = -\sigma_1 f_1 - \sigma_2 f_2 + \dfrac{\omega h}{3} (\sigma_K \sigma_1 f_1 - \sigma_K \sigma_2 f_2) \\
\sigma_\perp= \sigma_1 f_1 + \sigma_2 f_2 + \dfrac{\omega h}{3} (\sigma_1 \sigma_K f_1 - \sigma_2 \sigma_K f_2) \\
\sigma^{-1}_3 = \sigma^{-1}_1 f_1 + \sigma^{-1}_2 f_2 + \dfrac{\omega h}{3}(\sigma_1^{-1} \sigma_K f_1 - \sigma_2^{-1} \sigma_K f_2).
\label{sigmaeff}
\end{array}
\end{equation}
Furthermore, the expressions of effective $\ep_\text{eff}$, $\mu_\text{eff}$ are as follows
\begin{equation}
\begin{array}{l}
\ep_\parallel(\omega,\kpara) = \ep_1 f_1 + \ep_2 f_2 + \dfrac{\omega^2 h^2}{6} \mu_0 f_1 f_2 (\ep_1-\ep_2) (\ep_1 f_1-\ep_2 f_2) \left(1-\dfrac{{\boldsymbol k}^2}{{\omega}^2} \dfrac{\ep_1+\ep_2}{\mu_0\ep_1 \ep_2}\right)\\[2mm]
\ep_\perp(\omega,\kpara) = \ep_1 f_1 + \ep_2 f_2 + \dfrac{\omega^2 h^2}{6} \mu_0 f_1 f_2 (\ep_1-\ep_2) (\ep_1 f_1-\ep_2 f_2)  \\[2mm]
\mu_\parallel(\omega,\kpara) = \mu_0 - \dfrac{\omega^2 h^2}{6} \mu^2_0 f_1 f_2 (\ep_1-\ep_2) (f_1-f_2) \\[3mm]
\mu_\perp(\omega,\kpara)= \mu_0 - \dfrac{\omega^2 h^2}{6} \mu^2_0 f_1 f_2 (\ep_1-\ep_2) (f_1-f_2)\left(1-\dfrac{{\boldsymbol k}^2}{{\omega}^2} \dfrac{\ep_1+\ep_2}{\mu_0\ep_1 \ep_2}\right)  \\[3mm]
{\ep}^{-1}_3(\omega,\kpara) = {\ep}^{-1}_1f_1+{\ep}^{-1}_2f_2 - \dfrac{\omega^2 h^2}{6} \mu_0 f_1 f_2 (\ep_1-\ep_2) ({\ep}^{-1}_1f_1-{\ep}^{-1}_2f_2) \left(1-\dfrac{{\boldsymbol k}^2}{{\omega}^2} \dfrac{\ep_1+\ep_2}{\mu_0\ep_1 \ep_2}\right) \\[2mm]
{\mu}^{-1}_3(\omega,\kpara) = {\mu_0}^{-1} + \dfrac{\omega^2 h^2}{6} f_1f_2(\ep_1-\ep_2) (f_1-f_2). \\
\end{array}
\label{epmueff}
\end{equation}
All the effective parameters are frequency dependent and with spatial dispersion.
These expressions turn out to be equivalent to the ones reported in \cite{Rytov56,Raguin93,Yeh96}, where the effective refractive index $n_{\rm eff}$ is expanded using a power series of period-to-wavelength ratio $\Lambda/\lambda$. Taking equation (5)(s-polarized incidence is considered) in paper of Yeh \cite{Yeh96} as an example, the dispersion relation for a two-component layered medium is approximated by taking the fourth order of $O[(\Lambda/\lambda)^2]$ as
\begin{equation}
  {\it K}^2 + \beta^2 \approx (\dfrac{n_o \omega}{c})^2+\dfrac{(ab)^2(n_1^2-n_2^2)^2 (\omega/c)^4}{12 \Lambda^2}
\label{yeh}
\end{equation}
where $K$ and $\beta$ are the $z$ and $x$ components of the Bloch wave vector, $a$ and $b$ are the thicknesses of the alternating layers, $\Lambda=a+b$ is the period, $n_1$ and $n_2$ are the indices of refraction of the corresponding layers, $c$ the velocity of light in vacuum, and
\begin{equation}
  n_o^2=\dfrac{a}{\Lambda}n_1^2 + \dfrac{b}{\Lambda}n_2^2
\end{equation}
Comparing with the notations in our formula, we have
\begin{equation}
\begin{array}{llll}
  {\it K}=k_3,& \beta=\kpara, &  n_1^2=\ep_1,& n_2^2=\ep_2 \\
  a=h_1,& b=h_2,& \Lambda=h, & n_o^2=\dfrac{a}{\Lambda}n_1^2 + \dfrac{b}{\Lambda}n_2^2
\end{array}
\end{equation}
hence $f_1=a/\Lambda$, $f_2=b/\Lambda$, then (\ref{yeh}) is
\begin{equation}
  k_3^2+\kpara^2 \approx \dfrac{\omega^2}{c^2} (\ep_1 f_1+\ep_2f_2)+\dfrac{\omega^4}{12c^4}f_1^2 f_2^2 (\ep_1-\ep_2)^2
  \label{yeh2}
\end{equation}
which contains the terms of $\omega^2$ and $\omega^4$. On the other hand, for the effective medium in our HOH process, the dispersion relation of $k_3$ versus $\omega$ is
\begin{equation}
  k_3^2=\dfrac{\omega^2}{c^2} (\ep_\perp \mu_\parallel - K_\perp^2)-\dfrac{\mu_\parallel}{\mu_3}\kpara^2,\quad {\rm s-polarization}
\label{k3}
\end{equation}
Let us substitute the expressions of effective parameters (\ref{epmueff}) into (\ref{k3}), and collect the terms up to $\omega^4$ in the calculation process, and finally we obtain the same formula as in (\ref{yeh2}). Similar calculation applied to a p-polarized incident wave shows again our homogenization provides exactly the same effective index. Hence, it is stressed that the effective parameters $\ep_{\rm eff}$, $\mu_{\rm eff}$ and $K_{\rm eff}$ achieved from the HOH algorithm contain more information than the single effective index parameter in \cite{Raguin93,Yeh96,Rytov56}, e.g. the artificial magnetism and bianisotropy from periodic dielectrics which can not be seen in the derivation of refraction index.

In the asymptotic process, we have noticed that the magnetoelectric coupling comes from the odd order approximation while the artificial magnetism and high order corrections to permittivity emerge from the even order approximation in (\ref{BCH}). This can be explained in the following way: The matrix $M_m (m=1,2)$ for the dielectric layer is off-diagonal, the terms of odd order approximation usually contain odd commutators, hence a diagonal matrix will result, the components of which correspond to $K$. However, the terms of even order approximation contain even commutators, the resulting matrix is always off-diagonal. This introduces the artificial magnetism and high order corrections to permittivity.

Moreover, these results are fully consistent with descriptions in terms of spatial dispersion \cite{Landau84,Agranovich06} where, expanding the permittivity in power series of the wave vector, first order yields optical activity and second order magnetic response. The equivalence of these two descriptions (frequency and wave vector power series) is confirmed by considering a unit cell with a center of symmetry, for example a stack of three homogeneous layers (permittivity $\ep_m$ and thickness $h_m$, $m$ = 1, 2, 3) with $\ep_3=\ep_1$ and $h_3 =h_1$. Extending (\ref{BCH}) to the case $\exp[A] \exp[B] \exp[A] = \exp[Z]$ (see Section \ref{sec:bch}), it is found that $K_{\rm eff}$ = 0, and thus it is
retrieved that both bianisotropy and optical activity vanish in a medium with a center of symmetry \cite{Landau84}.

The present expansion in power series of frequency provides a new explanation for artificial magnetism and magnetoelectric coupling. Analytic expressions (\ref{epmueff}) of effective parameters can be used to analyze artificial properties. In particular, we show from (\ref{epmueff}) that: Artificial magnetism, previously proposed with high contrast \cite{Brien02,Chered06,Felbacq05}, can be obtained with arbitrarily low contrast; and bianisotropy, previously achieved in $\Omega$-composites \cite{Tretyakov07}, can be present in simple one-dimensional multilayers.
Note that, one can obtain more accurate asymptotic expressions for the effective parameters with more terms in (\ref{K}) and (\ref{epmueff}), by taking higher order approximation in (\ref{BCH}).

Although this homogenized system has been studied by \cite{Ramakrishna09}, these authors assumed some magnetism and bianisotropy for the periodic multilayered stack, whereas in our case bianisotropy and magnetism come from a homogenization process (one might say ex nihilo). Moreover these authors assumed that the bianisotropy matrix was diagonal, which is not the case in the present paper. It is to the best of our knowledge the first time these constitutive relations are derived, and we emphasize that the mathematical theorem invoked in this section (BCH formula) can be used to generalize our result to two dimensional and three dimensional periodic structures \cite{Lifante05}, such as woodpiles \cite{Gralak03}. In the sequel, we shall also investigate numerically the stop band properties of such a periodic stack of dielectrics, and draw some illuminating parallels with the seminal paper by Pendry \cite{Pendry04}.

\section{Extension of BCH formula for $m$ layers} \label{sec:bch}
In Section \ref{sec:hoh}, we have introduced our HOH algorithm for a periodic multilayered stack consisting of an alternation of two dielectric layers, wherein the BCH formula is implemented. In this section, we would like to investigate the extension of HOH to a stack with $m$ layers in a unit cell, correspondingly, a new form of BCH formula should be explored. We start with $m=3$, which means a multilayered stack with an alternation of three dielectric layers is considered, the thickness of each layer in a unit cell is $h_m$ with $m=1,2,3$, then the transfer matrix of one unit cell will be
\begin{equation}
T= T_1 T_2 T_3= \exp[i M_1 h_1]  \exp[i M_2 h_2]  \exp[i M_3 h_3] \,
\label{T1T2T3}
\end{equation}
rewritten in a more general form, e.g.
\begin{equation}
  \exp[Z]=\exp[A_1]\exp[A_2]\exp[A_3]
\label{bch3}
\end{equation}
defines a product of three exponential functions. Obviously, it can be solved by an iteration of BCH formula. First, we suppose
\begin{equation}
  \exp[A]=\exp[A_1] \exp[A_2]
\end{equation}
and $A$ can be derived through equation (\ref{BCH})
\begin{equation}
  A = A^{(0)} + A^{(1)} + A^{(2)} + A^{(3)} + \cdots
\end{equation}
The term $A^{(i)}$ represents the $i^{th}$ order approximation, and
\begin{equation}
\begin{array}{l}
  A^{(0)} = A_1+A_2\, , \\[1mm]
  A^{(1)}=\dfrac{1}{2} \llbracket A_1,A_2\rrbracket\, , \\[3mm]
  A^{(2)} =\dfrac{1}{12} \llbracket A_1, \llbracket A_1,A_2 \rrbracket \rrbracket-\dfrac{1}{12} \llbracket A_2, \llbracket A_1,A_2 \rrbracket \rrbracket\, , \\
  \cdots
\end{array}
\end{equation}
then (\ref{bch3}) turns to be
\begin{equation}
\exp[Z]=\exp[A]\exp[A_3]\,.
\end{equation}
Using the BCH formula for $Z$:
\begin{equation}
\begin{array}{ll}
  Z&=\log{(\exp[A] \exp[A_3])}=A+A_3+\dfrac{1}{2} \llbracket A,A_3\rrbracket  \\[3mm]
  &+ \dfrac{1}{12} \llbracket A, \llbracket A, A_3 \rrbracket \rrbracket-\dfrac{1}{12} \llbracket A_3, \llbracket A, A_3 \rrbracket \rrbracket
  - \dfrac{1}{24} \llbracket A_3, \llbracket A, \llbracket A, A_3 \rrbracket \rrbracket \rrbracket + \cdots
\end{array}
\end{equation}
we suppose $Z=Z^{(0)}+Z^{(1)}+Z^{(2)}+\cdots$, where $Z^{(m)}$ is the $m^{th}$ order approximation for $Z$. The zeroth order $Z^{(0)}$ is simply the sum of $A_1$, $A_2$ and $A_3$,
\begin{equation}
 Z^{(0)}= A^{(0)} + A_3 =A_1+A_2+A_3\,.
\label{z0}
\end{equation}
The first order including single commutator of these three matrices $A_1$, $A_2$ and $A_3$ is
\begin{equation}
 Z^{(1)}=A^{(1)}+\dfrac{1}{2} \llbracket A^{(1)},A_3\rrbracket = \dfrac{1}{2} \llbracket A_1,A_2\rrbracket +\dfrac{1}{2} \llbracket A_1+A_2,A_3\rrbracket\,.
\label{z1}
\end{equation}
and the second order including double commutators is
\begin{equation}
\begin{array}{ll}
Z^{(2)}&= A^{(2)}+\dfrac{1}{2} \llbracket A^{(2)},A_3\rrbracket + \dfrac{1}{12} \llbracket A^{(1)}, \llbracket A^{(1)}, A_3\rrbracket \rrbracket - \dfrac{1}{12} \llbracket A_3,
\llbracket A^{(1)}, A_3\rrbracket \rrbracket  \\[3mm]
&= \dfrac{1}{12} \llbracket A_1, \llbracket A_1,A_2 \rrbracket \rrbracket-\dfrac{1}{12} \llbracket A_2, \llbracket A_1,A_2 \rrbracket \rrbracket
+ \dfrac{1}{4} \llbracket \llbracket A_1,A_2\rrbracket,A_3\rrbracket  \\[3mm]
&+ \dfrac{1}{12} \llbracket A_1+A_2, \llbracket A_1+A_2, A_3\rrbracket \rrbracket - \dfrac{1}{12} \llbracket A_3, \llbracket A_1+A_2, A_3\rrbracket \rrbracket \,.
\end{array}
\label{z2}
\end{equation}
A similar algorithm holds for the third and higher orders, which will not be further explored here.

Since the BCH formula for (\ref{bch3}) has been derived, one can easily realize the homogenization for a multilayered stack with an alternation of three dielectric layers. Here we assume that the third layer of the unit cell is identical to the first layer, i.e. $\ep_3=\ep_1$ and $h_3=h_1$, as well as $M_3=M_1$;
taking equation (\ref{z1}) with $A_1 = A_3$, we have $Z^{(1)}=0$. It should be noted that all the odd orders of approximation in the HOH asymptotics vanish, which can be attributed to the center symmetric property of the structure \cite{Boris08}. In contrast, even orders rule the approximation process in that case. Applying the formulae (\ref{z0})-(\ref{z2}) to (\ref{T1T2T3}), one deduces the expressions for the effective parameters at 2nd order approximation
\begin{equation}
\begin{array}{l}
\ep_\parallel=2\ep_1 f_1 + \ep_2 f_2 -\dfrac{\omega^2 h^2}{3} \mu_0 f_1 f_2 (\ep_1-\ep_2) (\ep_1 f_1+\ep_2 f_2)\left(1-\dfrac{{\boldsymbol k}^2}{\omega^2}\dfrac{\ep_1+\ep_2}{\mu_0\ep_1\ep_2}\right) \\
\ep_\perp=2\ep_1 f_1 + \ep_2 f_2 -\dfrac{\omega^2 h^2}{3} \mu_0 f_1 f_2 (\ep_1-\ep_2) (\ep_1 f_1+\ep_2 f_2)  \\
\mu_\parallel =\mu_0 + \dfrac{\omega^2 h^2}{3}\mu_0^2 f_1 f_2 (\ep_1-\ep_2) (f_1+f_2)  \\[3mm]
\mu_\perp=\mu_0 + \dfrac{\omega^2 h^2}{3}\mu_0^2 f_1 f_2 (\ep_1-\ep_2) (f_1+f_2)\left(1-\dfrac{{\boldsymbol k}^2}{\omega^2}\dfrac{\ep_1+\ep_2}{\mu_0\ep_1\ep_2}\right)  \\[2mm]
\ep_3^{-1}=2{\ep}^{-1}_1f_1 +{\ep}^{-1}_2f_2+\dfrac{\omega^2 h^2}{3} \mu_0 f_1 f_2 (\ep_1-\ep_2)({\ep}^{-1}_1f_1+{\ep}^{-1}_2f_2)
\left(1-\dfrac{{\boldsymbol k}^2}{\omega^2}\dfrac{\ep_1+\ep_2}{\mu_0\ep_1\ep_2}\right)  \\[2mm]
\mu_3^{-1}=\mu_0^{-1}-\dfrac{\omega^2 h^2}{3} f_1 f_2 (\ep_1-\ep_2) (f_1+f_2)  \\[2mm]
K_\perp = K_\parallel = 0 \, .
\label{epmueff2}
\end{array}
\end{equation}
The effective bianisotropy $K_{\rm eff}$ is equal to zero since $Z^{(2p+1)}=0$, and only the artificial magnetism and high order corrections to the permittivity persist. Similar calculation can be applied to higher order approximation, e.g. the expressions of these parameters in 4th order approximation under a normal incidence is discussed in \cite{yan12}.

So far, we have discussed the HOH asymptotic for a multilayered stack consisting of an alternation of two layers, as well as three layers; and the BCH formula has been also amended correspondingly. If we extend this asymptotic procedure to a more general case, i.e. we consider a stack with an alternation of $m (\geq 3)$ layers, then the transfer matrix of a unit cell becomes
\begin{equation}
  T=\exp[Z]=\prod\limits_{i=1}^{m} \exp[A_i] \, .
 \label{zm}
\end{equation}
Once again, tedious iteration of BCH in (\ref{zm}) can produce all the formulae for different orders of approximation. Here, we just list the formulae from zeroth order to second order approximation:
\begin{equation}
\begin{array}{c}
  Z^{(0)}=\sum\limits_{i=1}^{m} A_i \, , \\
  Z^{(1)}=\dfrac{1}{2}\sum\limits_{i=2}^{m}\, \llbracket \sum\limits_{j=2}^{i}A_{j-1},A_i \rrbracket \, , \\
  Z^{(2)}=\dfrac{1}{4}\sum\limits_{i=3}^{m} \llbracket \llbracket \sum\limits_{j=1}^{i-2}A_j, A_{i-1} \rrbracket, \sum\limits_{i}^{m}A_i \rrbracket \\ +\dfrac{1}{12}\sum\limits_{i=2}^{m} \left(\llbracket \sum\limits_{j=1}^{i-1} A_j, \llbracket \sum\limits_{j=1}^{i-1} A_j, A_i \rrbracket\rrbracket - \llbracket A_i, \llbracket \sum\limits_{j=1}^{i-1} A_j, A_i \rrbracket\rrbracket\right)\, . \\
\end{array}
\label{bchm}
\end{equation}
These formulae can be checked by taking $m=3$ and then compare with equations (\ref{z0})-(\ref{z2}). Apart from an iteration of BCH formula, another method to obtain the approximation for $Z$ would be to expand each exponential function by Taylor series, and collect the terms with same order, which will not be further discussed in this paper.
%----------------------------------------------------------------------------------------------------------------------------
\section{Corrector for HOH asymptotics} \label{sec:corc}

In this section, we would like to introduce the corrector for the asymptotic error in the HOH algorithm, where a structure with thickness constant with respect to the frequency or the wavelength should be considered. We start with a structure consisting of two dielectric layers as shown in figure \ref{figure2}(a), where the parameters are $\ep_1$, $\ep_2$, and thicknesses $\tilde{h}_1$, $\tilde{h}_2$ satisfying $\tilde{h}_1+\tilde{h}_2 \approx \lambda$. In order to obtain a homogeneous effective medium for such a structure, a geometric reconstruction is implemented here by dividing the structure into $n$ times smaller unit cells, figure \ref{figure2}(b) shows the new construction when $n=2$, the thicknesses of two layers being $\tilde{h}_1/2$ and $\tilde{h}_2/2$, respectively. It is noted that the thickness of each layer in a unit cell is decreasing in proportion to an increasing $n$, as shown in figure \ref{figure2}(c). When $n$ tends to a large enough constant, the thickness of the unit cell $\tilde{h}_i/n$ will be much smaller than the wavelength of light, and a homogeneous effective medium can be achieved as shown in figure \ref{figure2}(d).
\begin{figure}[!htb]
    \centering
    \includegraphics[scale=0.6]{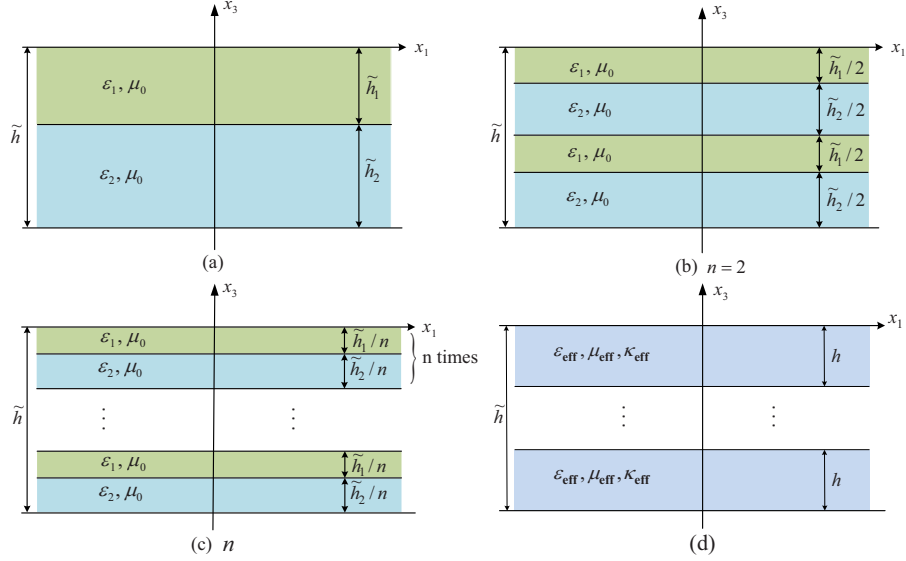}
    \caption{Schematic diagram of the homogenization algorithm for a structure with constant thickness: (a) A structure consisting of two homogeneous dielectric layers of permittivities $\varepsilon_1$, $\varepsilon_2$ and thicknesses $\tilde{h}_1$, $\tilde{h}_2$; (b)-(c) A reconstruction of (a) by dividing the unit cell into $n$ times smaller unit cells where $n=2$ in (b), and a large $n$ in (c); (d) An effective medium with thickness $\tilde{h}$.}
\label{figure2}
\end{figure}

The transfer matrices of the periodic medium and the effective homogenized medium should satisfy:
\begin{equation}\label{liu3}
\big\{ \exp[i M_2 \tilde{h}_2/n] \exp[i M_1 \tilde{h}_1/n] \big\}^n
\simeq \exp[i M_\text{eff} \tilde{h}],\, \, \, {\rm for\,\, large}\, n  .
\end{equation}
The BCH formula is still central to solve this problem, hence we recall its statement:
\begin{align}\label{campbell}
   \exp \big[ \, A_1 \, \big] \exp \big[ \, A_2 \, \big] &=
   \exp \big[ \,A_1 + A_2 + \dfrac{1}{2}\llbracket A_1, A_2 \rrbracket + \dfrac{1}{12}\llbracket  A_1 , \llbracket A_1, A_2 \rrbracket \rrbracket  \notag\\
   &- \dfrac{1}{12} \llbracket  A_2 , \llbracket A_1, A_2 \rrbracket \rrbracket + \cdots \big] \, .
\end{align}
Let us take the first order estimate in (\ref{campbell}). This leads to the following error estimate:
\begin{proposition}
If $A_1$ and $A_2$ in (\ref{campbell}) are bounded by $\left\| A \right\| / 2$, then
\begin{equation}\label{1order1}
\left\| \left\{
\exp \left[ A_1 / n \right] \exp \left[ A_2 / n \right]
\right\}^n \! -
\exp \left[ \left\langle A \right\rangle+ L_A / n
\right] \right\| \leq \frac{a}{n^2} \, ,
\end{equation}
with
\begin{equation}
\begin{array}{c}
\vspace*{2mm}
\left\langle A \right\rangle = A_1 + A_2 \, , \quad \quad
L_A = \llbracket A_1, A_2 \rrbracket / 2 \,,
\vspace*{2mm}\\
a = \displaystyle{\frac {\left\| A \right\|^3} {3}}
\exp \left[ \, 3 \, \left\| A \right\| \right] \,
\exp \left[ \, \left\| A \right\|^2  \right] \,.
\end{array}
\label{1order2}
\end{equation}
\end{proposition}

\begin{proof}
Let $S_n$ and $T_n$ be defined by
\begin{equation}\label{1order3}
\begin{array}{c}
   S_n = \exp\left[ \left\langle A \right\rangle /n + L_A /n^2 \right]\,,
\vspace*{2mm}\\
   T_n = \exp\left[ A_1/n \right] \exp\left[ A_2/n \right]\,.
\end{array}
\end{equation}
One can write
\begin{equation}
S_n^n - T_n^n = \sum_{p = 1}^{n} S_n^{p-1}
\left( S_n - T_n \right) T_n^{n-p} \, ,
\label{1order4}
\end{equation}
so that
\begin{equation}
\left\| S_n^n - T_n^n \right\| \leq
\sum_{p = 1}^{n}
\left\| S_n \right\|^{p-1}
\left\| S_n - T_n \right\|
\left\| T_n \right\|^{n-p} \, .
\label{1order5}
\end{equation}
Straightforward upper bounds for $S_n$ and $T_n$ are:
\begin{equation}\label{1order6}
\begin{array}{c}
   \left\|S_n\right\|  \leq  \exp\left[ \left\| A \right\| /n \right]
   \exp\left[ \left\| A \right\|^2 / n^2 \right]\,,
   \vspace*{2mm}\\
   \left\|T_n\right\|  \leq  \exp\left[ \left\| A \right\| /n \right] \,.
\end{array}
\end{equation}
Then developing exponential functions in (\ref{1order3}) as series,
\begin{equation}\label{1order7}
\begin{array}{c}
   S_n  = 1 + \left( \left\langle A \right\rangle /n + L_A /n^2 \right)
   +\left( \left\langle A \right\rangle /n + L_A /n^2 \right)^2/2+ \cdots \,
   \vspace*{2mm}\\
   T_n  = \left[\, 1 + A_1 / n + A_1^2 / (2 n^2) + \cdots \, \right]
  \left[ \, 1 + A_2 / n + A_2^2 / (2 n^2) + \cdots \, \right] \,
\end{array}
\end{equation}
one has
\begin{equation}
\left\| S_n - T_n \right\| \leq
\frac{\left\| A \right\|^3}{3 n^3}
\exp\left[ 2 \left\| A \right\|\right]
\exp\left[ \left\| A \right\|^2 \right] \, .
\label{1order8}
\end{equation}
Substituting (\ref{1order6}) and (\ref{1order8}) into (\ref{1order5}) provides the results (\ref{1order1})-(\ref{1order2}).
\end{proof}

In periodic classical homogenization \cite{Bensoussan78,Bakhalov89},
only the leading order term is kept in the asymptotic procedure, hence
the corrector is of order $1/n$. Here, we find that:
\begin{equation}\label{0order11}
\left\| \left\{
\exp \left[ A_1 / n \right] \exp \left[ A_2 / n \right]
\right\}^n \! -
\exp \left[ \left\langle A \right\rangle\right] \right\| \leq \frac{b}{n} \, ,
\end{equation}
with
\begin{equation}
\vspace*{2mm}
\left\langle A \right\rangle = A_1 + A_2 \, , \quad \quad
b =\left\| A \right\|^2
\exp \left[ 2\left\| A \right\| \right] \,.
\end{equation}
We therefore emphasize that our iterative procedure amounts to keeping more and more terms in the asymptotic expansion of
classical homogenization and thus improves the order of the corrector of classical homogenization. Similar ideas have been implemented in the high-frequency homogenization recently developed by Craster et al. \cite{Craster10}, however with no correctors being derived therein.
It would be also interesting to see how randomness would affect our correctors: at order zero, the corrector is known to vary between $n^{-1/2}$ and $n^{-1}$ \cite{Bourgeat99}.

\noindent Furthermore,using the same proof process, we can also obtain the second order estimate of the limits (\ref{campbell}).
\begin{proposition}
If $A_1$ and $A_2$ in (\ref{campbell}) are bounded by $\left\| A \right\| / 2$, then
\begin{equation}\label{2order1}
\left\| \left\{
\exp \left[ A_1 / n \right] \exp \left[ A_2 / n \right]
\right\}^n \! -
\exp \left[ \left\langle A \right\rangle+ L_A / n+R_A/n^2
\right] \right\| \leq \frac{a}{n^3} \, ,
\end{equation}
with
\begin{equation}
\begin{array}{c}
\vspace*{2mm}
\left\langle A \right\rangle = A_1 + A_2 \, , \quad \quad
L_A = \llbracket A_1, A_2 \rrbracket / 2 \, ,\vspace*{2mm}\\
R_A = \dfrac{1}{12}\llbracket  A_1 , \llbracket A_1, A_2 \rrbracket \rrbracket  - \dfrac{1}{12} \llbracket  A_2 , \llbracket A_1, A_2 \rrbracket \rrbracket \vspace*{2mm}\\
a = \displaystyle{\frac{\left\| A \right\|^4}{4}}
\exp \left[ \, 3 \, \left\| A \right\| \right] \,
\exp \left[ \, \left\| A \right\|^2  \right] \,
\exp \left[ \, \left\| A \right\|^3  \right] \,.
\end{array}
\label{2order2}
\end{equation}
\end{proposition}

\begin{proof}
Let $S_n$ and $T_n$ be defined by
\begin{equation}\label{2order3}
\begin{array}{c}
   S_n = \exp\left[ \left\langle A \right\rangle /n + L_A /n^2 + R_A/n^3 \right]\,
\vspace*{2mm}\\
   T_n = \exp\left[ A_1/n \right] \exp\left[ A_2/n \right]\,.
\end{array}
\end{equation}

Upper bounds for $S_n$ and $T_n$ are straightforward:
\begin{equation}\label{2order4}
\begin{array}{c}
   \left\|S_n\right\|  \leq  \exp\left[ \left\| A \right\| /n \right]
   \exp\left[ \left\| A \right\|^2 / n^2 \right]
   \exp\left[ \left\| A \right\|^3 / n^3 \right]\,
   \vspace*{2mm}\\
   \left\|T_n\right\|  \leq  \exp\left[ \left\| A \right\| /n \right]\,.
\end{array}
\end{equation}

And developing the exponential function in (\ref{2order3}) as a series ,
\begin{equation}\label{2order5}
\begin{array}{c}
   S_n  = 1 + \left( \left\langle A \right\rangle /n + L_A /n^2 + R_A/n^3 \right)
   +\left( \left\langle A \right\rangle /n + L_A /n^2 + R_A/n^3 \right)^2/2+ \cdots
   \vspace*{2mm}\\
   T_n  = \left[\, 1 + A_1 / n + A_1^2 / (2 n^2) + \cdots \, \right]
  \left[ \, 1 + A_2 / n + A_2^2 / (2 n^2) + \cdots \, \right] \, ,
\end{array}
\end{equation}

one has
\begin{equation}
\left\| S_n - T_n \right\| \leq
\frac{\left\| A \right\|^4}{4 n^4}
\exp\left[ 2 \left\| A \right\|\right]
\exp\left[ \left\| A \right\|^2\right]
\exp\left[ \left\| A \right\|^3 \right]\, .
\label{2order6}
\end{equation}
Substituting (\ref{2order4}), (\ref{2order6}) into (\ref{1order5}) leads to (\ref{2order1})-(\ref{2order2}).
\end{proof}
It is noted that as the approximation order increases, the speed of the convergence defined by the difference between transfer matrices of multilayers and effective medium increases by a factor $1/n$, hence it seems natural to conjecture that for the higher order approximation, the estimate between the transfer matrices of multilayers and the effective medium will be much more accurate with an error of $1/n^p$, with $p$ the order taken in HOH approximation process.

Similarly, the asymptotic corrector can be applied to the stack with $m$ layers. Here, we explore the corrector for HOH approximation in a multilayered stack with three layers, where the permittivities are $\ep_1$, $\ep_2$, $\ep_3$ and thicknesses are $\tilde{h}_1/n$, $\tilde{h}_2/n$, $\tilde{h}_3/n$. Applying the same geometric reconstruction as shown in figure \ref{figure2}, the equivalent relation between the transfer matrices of the stack and effective medium is
\begin{equation}\label{liu33}
\big\{ \exp[i M_1 \tilde{h}_1/n] \exp[i M_2 \tilde{h}_2/n] \exp[i M_3 \tilde{h}_3/n] \big\}^n
%\underset{n}
\simeq \exp[i M_\text{eff} \tilde{h}] \, .
\end{equation}
According to equations (\ref{T1T2T3}) and (\ref{z0})-(\ref{z2}), the BCH leads to
\begin{align}
   \exp \big[ \, A_1 \, \big] \exp \big[ \, A_2 \, \big]\exp \big[ \, A_3 \, \big] &=
   \exp \big[ \, A_1 + A_2 + A_3+ (A_1 A_2 - A_2 A_1) / 2  \notag \\
   & + (A_1 A_3 - A_3 A_1) / 2 +(A_2 A_3 - A_3 A_2) / 2 +\cdots \big] \, ,
\label{corr3}
\end{align}
Taking the zeroth order approximation as an example, we can state
\begin{proposition}
If $A_1$, $A_2$ and $A_3$ in (\ref{corr3}) are bounded by $\left\| A \right\| / 3$, then
\begin{equation}\label{0order31}
\left\| \left\{
\exp \left[ A_1 / n \right] \exp \left[ A_2 / n \right] \exp \left[ A_3 / n \right]
\right\}^n \! -
\exp \left[ \left\langle A \right\rangle
\right] \right\| \leq \frac{a}{n} \, ,
\end{equation}
with
\begin{equation}
\begin{array}{c}
\vspace*{2mm}
\left\langle A \right\rangle = A_1 + A_2 +A_3 \, ,
\vspace*{2mm}\\
a = {\left\| A \right\|^2}
\exp \left[ \, 2\, \left\| A \right\| \right] \,.
\end{array}
\label{0order32}
\end{equation}
\end{proposition}

\begin{proof}
Let $S_n$ and $T_n$ be defined by
\begin{equation}\label{0order33}
\begin{array}{c}
   S_n = \exp\left[ \left\langle A \right\rangle /n \right]\,
\vspace*{2mm}\\
   T_n = \exp\left[ A_1/n \right] \exp\left[ A_2/n \right]\exp\left[ A_3/n \right]\,.
\end{array}
\end{equation}
Equations (\ref{1order4}) and (\ref{1order5}) remain valid and upper bounds for $S_n$ and $T_n$ are
\begin{equation}\label{0order34}
\begin{array}{c}
   \left\|S_n\right\|  \leq  \exp\left[ \left\| A \right\| /n \right]\,,
   \vspace*{2mm}\\
   \left\|T_n\right\|  \leq  \exp\left[ \left\| A \right\| /n \right]\,.
\end{array}
\end{equation}
Then developing the exponential function as a series,
\begin{equation}\label{0order35}
\begin{array}{c}
   S_n  = 1 + \left( \left\langle A \right\rangle /n  \right)
   +\left( \left\langle A \right\rangle /n  \right)^2/2+ \cdots \,
   \vspace*{2mm}\\
   T_n  = \left[\, 1 + A_1 / n + A_1^2 / (2 n^2) + \cdots \, \right] \left[ \, 1 + A_2 / n + A_2^2 / (2 n^2) + \cdots \, \right]\vspace*{2mm} \\
  \left[ \, 1 + A_3 / n + A_3^2 / (2 n^2) + \cdots \, \right] \, ,
\end{array}
\end{equation}
one has
\begin{equation}
\left\| S_n - T_n \right\| \leq
\frac{\left\| A \right\|^3}{n^2}
\exp\left[ 2 \left\| A \right\|\right] \, .
\label{0order36}
\end{equation}
Substituting (\ref{0order34}) and (\ref{0order36}) into (\ref{1order5}) leads to (\ref{0order31}) and (\ref{0order32}).
\end{proof}
The proof indicates that the corrector is in order of $n^{-1}$ when taking the classical homogenization (zeroth order approximation) for a multilayered stack with an alternation of three layers, this can be adopted for the $m$ layers case. Correctors for higher order approximation, as well as for a multilayered stack consisting of an alternation of $m$ layers can be obtained by the same algorithm.

\section{Numerical calculations: Dispersion law and transmission curves}
\label{sec:num}
In this section, we would like to numerically investigate the asymptotic degree between the multilayered stack and its effective medium obtained from HOH algorithm, where the dispersion law and transmission curves are explored. According to the previous analysis, the transfer matrix defined by the exponential function of matrix $M$ is analytic, and it can be expanded as a Taylor series, taking the effective transfer matrix as an example,
\begin{equation}
T_{\rm eff} = \exp(-iM_{\rm eff}h)
=\sum_{p=0}^{\infty} (-i)^{2p} \dfrac{M_{\rm eff}^{2p}h^{2p}}{(2p)!}+\sum_{p=0}^{\infty} (-i)^{2p+1} \dfrac{M_{\rm eff}^{2p+1}h^{2p+1}}{(2p+1)!} \,.
\label{T}
\end{equation}
Considering a s-polarized incident wave, the column vector in (\ref{F}) is defined by
\begin{equation}
  {\widehat{\text{F}}}=[\, {\widehat{\text{E}}}_\perp, {\widehat{\text{H}}}_\parallel \,]^{\rm T}
\end{equation}
The matrix $M_{\rm eff}$ in (\ref{Meff}) is a 2 by 2 matrix, and
\begin{equation}
M_{\rm eff}^2 =
 \left[ \begin{array}{cc} \vspace*{1mm}
 i\omega K_\perp & -\omega \mu_\parallel \\
 -\omega \ep_\perp +\dfrac{\kpara^2}{\omega \mu_3} & -i\omega K_\perp
 \end{array} \right]^2 = k_{\rm eff}^2
  \left[ \begin{array}{cc} \vspace*{1mm}
 1 & 0\\
 0 & 1\\
 \end{array} \right]\,.
\label{Mte2}
\end{equation}
where
\begin{equation}
k_{\rm eff}^2=\omega^2 (\ep_\perp \mu_\parallel- K_\perp^2)-\dfrac{\mu_\parallel}{\mu_3} \kpara^2 \, .
\label{keff}
\end{equation}
Plugging (\ref{Mte2}) into (\ref{T}) and considering Taylor series of the $\sin-\cos$ functions
\begin{equation}
    \sin(A)=\sum_{n=0}^{\infty} \dfrac{(-1)^n}{(2n+1)!} A^{2n+1} \, , \quad
    \cos(A)=\sum_{n=0}^{\infty} \dfrac{(-1)^n}{(2n)!} A^{2n} \, ,
\label{sincos}
\end{equation}
we obtain
\begin{equation}
\begin{array}{ll}
T_{\rm eff}&=\cos(k_{\rm eff} h)-i\dfrac{M_{\rm eff}}{k_{\rm eff}} \sin(k_{\rm eff} h) \\
&=\left[ \begin{array}{cc} \vspace*{1mm}
 \cos(k_{\rm eff} h)+\omega K_\perp \dfrac{\sin(k_{\rm eff} h)}{k_{\rm eff}} &
 i\omega \mu_\parallel \dfrac{\sin(k_{\rm eff} h)}{k_{\rm eff}}\\
 i(\omega\ep_\perp-\dfrac{\kpara^2}{\omega}\dfrac{1}{\mu_3} )\dfrac{\sin(k_{\rm eff} h)}{k_{\rm eff}} &
\cos(k_{\rm eff} h)-\omega K_\perp \dfrac{\sin(k_{\rm eff}h)}{k_{\rm eff}}\\
 \end{array} \right] \,.
\end{array}
\label{Tte}
\end{equation}
Similarly, the transfer matrix of the dielectric layer $\mathcal{L}_m$ is
\begin{equation}
T_m=\left[\begin{array}{cc}
\cos{(\beta_m h_m)} & i\omega \mu_0 \dfrac{\sin(\beta_m h_m)}{\beta_m}\\
i(\omega\ep_m-\dfrac{\kpara^2}{\omega}\dfrac{1}{\mu_0})\dfrac{\sin(\beta_m h_m)}{\beta_m} & \cos{(\beta_m h_m)}
\end{array}\right], \quad \beta_m^2=\omega^2\ep_m \mu_m-\kpara^2 .
\label{tp}
\end{equation}
The transfer matrix $T$ of the unit cell consisting of two dielectric layers is derived from the above expression as
\begin{equation}
  T=T_2 T_1 \,.
\label{Tunit}
\end{equation}

\subsection{Dispersion law}
A general expression for the dispersion law in a periodic structure is defined by the trace of transfer matrix $T$ of a single period \cite{Keldysh88,Ivchenko91,Deutsch95}. Since the eigenvalues and eigenvectors of $T$ (and thus any power of $T$) are the Bloch wave vectors and Bloch states of the periodic structure (in the limit of infinite $n$), further physical insight can be achieved in the single period matrix $T$. Hence, for a multilayered stack with two layers, we have
\begin{equation}
  {\rm tr(T)}/2 = \cos{(\beta_1 h_1)}\cos{(\beta_2 h_2)}-\dfrac{1}{2}(\dfrac{\beta_2}{\beta_1}+\dfrac{\beta_1}{\beta_2})\sin{(\beta_1 h_1)}\sin{(\beta_2 h_2)}
\label{trTstack}
\end{equation}
while for the effective medium,
\begin{equation}
  {\rm tr(T_{\rm eff})}/2 = \cos(k_{\rm eff} h)\,.
\label{trTeff}
\end{equation}
Here, $k_{\rm eff}$ defined by (\ref{keff}) can be obtained by substituting the expressions of the effective permittivity, permeability and bianisotropy, which are derived from the HOH algorithm.

Considering a normal incident plane wave ($\kpara=0$), the s-polarization and p-polarization coincide, since $\ep_\parallel=\ep_\perp$, $\mu_\parallel=\mu_\perp$, $K_\parallel=K_\perp$. Hence we take a s-polarized incident wave as an example, and assume the two dielectric layers of the stack are Glass and Silicon, respectively; the relative permittivities are $\ep_1=2$, $\ep_2=12$ and the filling fraction are $f_1=0.8$, $f_2=0.2$. For the sake of illustration, in the HOH algorithm, we take the 3rd, 7th and 19th order approximation for the effective medium. The curves of the effective permittivity, permeability and bianisotropy in 19th order approximation versus normalized frequency are depicted in figure \ref{figure3}(a), where $\hat \omega=\omega h/c$, the expressions of these effective parameters are omitted here to save space. It is observed that all three curves are increasing along with the frequency, wherein the effective permeability (dash red line) has values greater than 1, and effective bianisotropy (dotted-dash blue line) is non-vanishing. In other words, artificial magnetism and bianisotropy can be achieved from dielectrics through HOH asymptotics, as it has been predicted in the theoretical analysis of Section \ref{sec:hoh}.
\begin{figure}[!htb]
    \centering
    \subfigure[]{\label{fig3a}
    \includegraphics[width=0.48\textwidth]{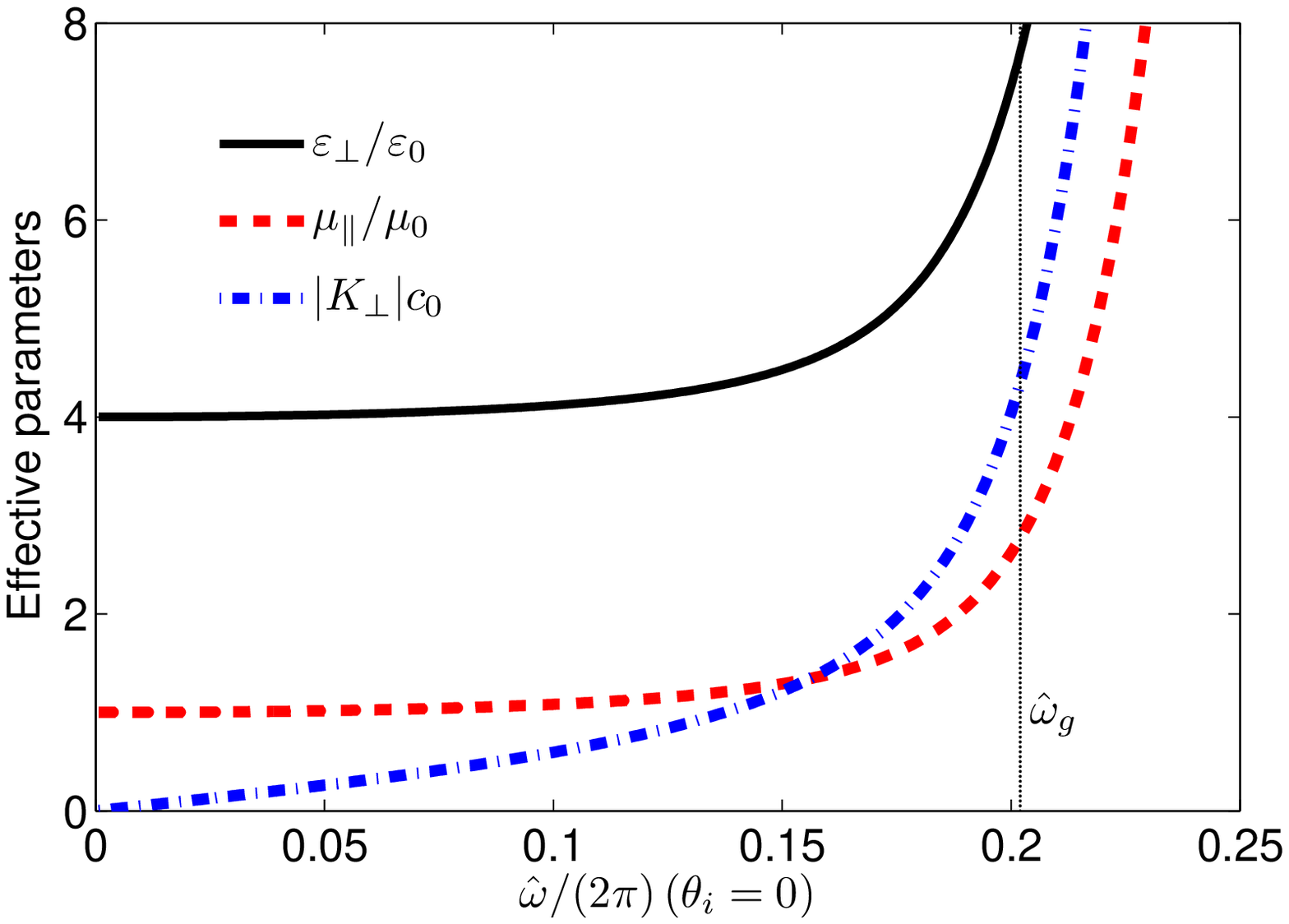}}
    \subfigure[]{\label{fig3b}
    \includegraphics[width=0.48\textwidth]{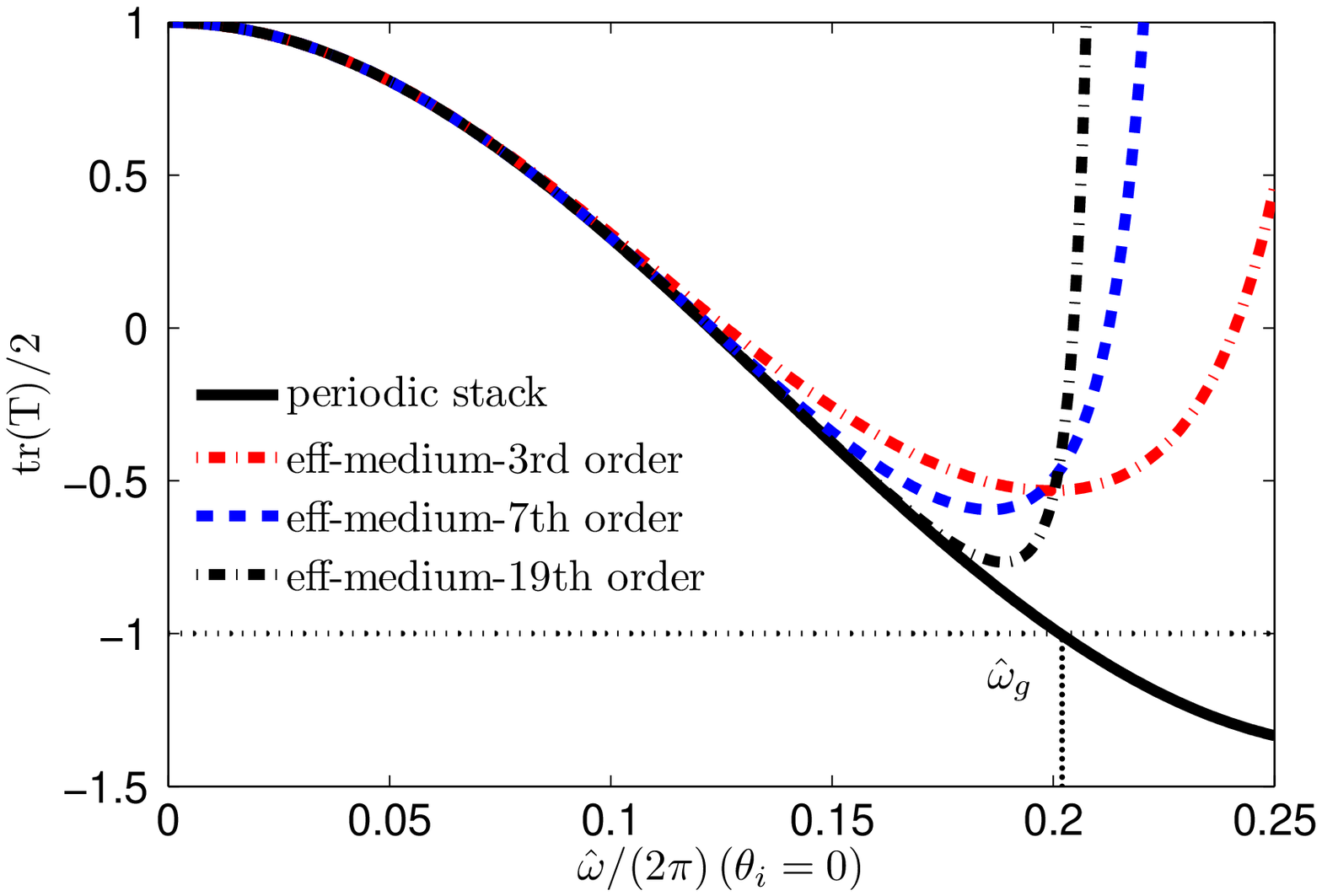}}
    \caption{(a) The curves of effective permittivity (solid black line), permeability (dash red line) and bianisotropy (dotted-dash blue line) in 19th order approximation; (b) Dispersion laws of the multilayered stack (solid black line) and its effective medium in 3rd order (dotted red line), 7th order (dash blue line) and 19th order (dotted-dash black line). HOH breaks down at $\hat{\omega}_g=0.202$, lower edge of the first stop band.}
     \label{figure3}
\end{figure}

Figure \ref{figure3}(b) shows the dispersion law of the stack, as well as that of the effective medium in different order approximations, which is obtained by substituting the expressions of the effective parameters into $k_{\rm eff}$ of (\ref{keff}) and then (\ref{trTeff}). The lower edge of the first stop band is denoted by $\hat \omega_g = 0.202$, where ${\rm tr(T)}/2=-1$ \cite{Lekner94}. It should be noted that a good agreement between the dispersion laws of the stack and effective medium can be observed at the lower frequency band, and the asymptotics of these two curves can be improved by taking higher orders of approximation, e.g. 19th order approximation (dotted-dash line) shows a better asymptotic with the multilayers (solid line) from zero frequency (quasi-static limit) up to a normalized frequency around $0.18$.

Moreover, an oblique incident plane wave in s-polarization as well as in p-polarization are also analyzed, wherein the same parameters of the stack are taken and the incident angle is $\theta_i=30^\circ$, the dispersion laws of the multilayered stack and the effective medium are shown in figure \ref{figure4}. The 3rd and 7th order approximations are considered in the HOH algorithm, similarly, asymptotics improve in conjunction with higher orders of approximation for both s- and p-polarized waves.
\begin{figure}[!htb]
    \centering
    \subfigure[s-polarization]{\label{fig4a}
    \includegraphics[width=0.48\textwidth]{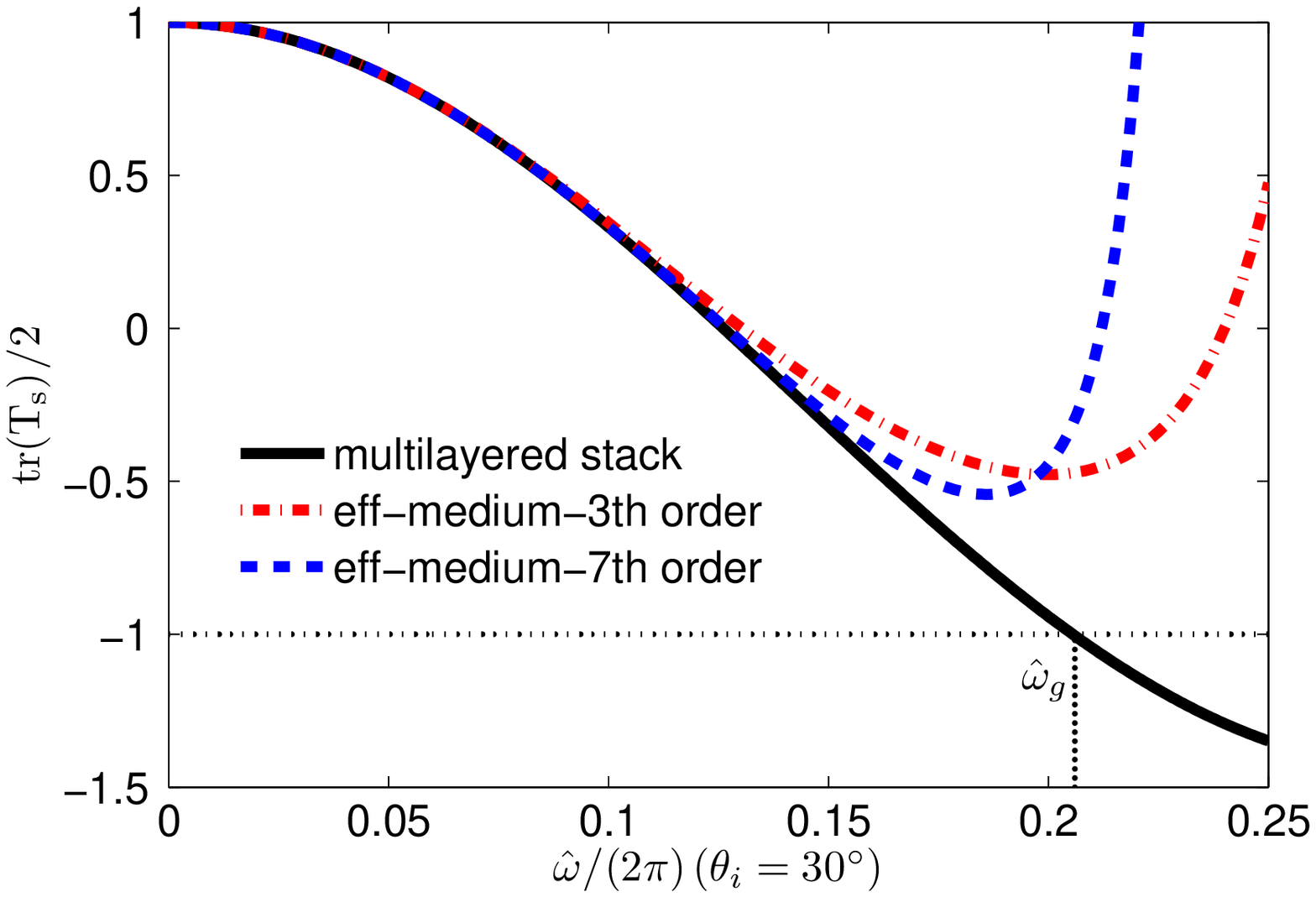}}
    \subfigure[p-polarization]{\label{fig4b}
    \includegraphics[width=0.48\textwidth]{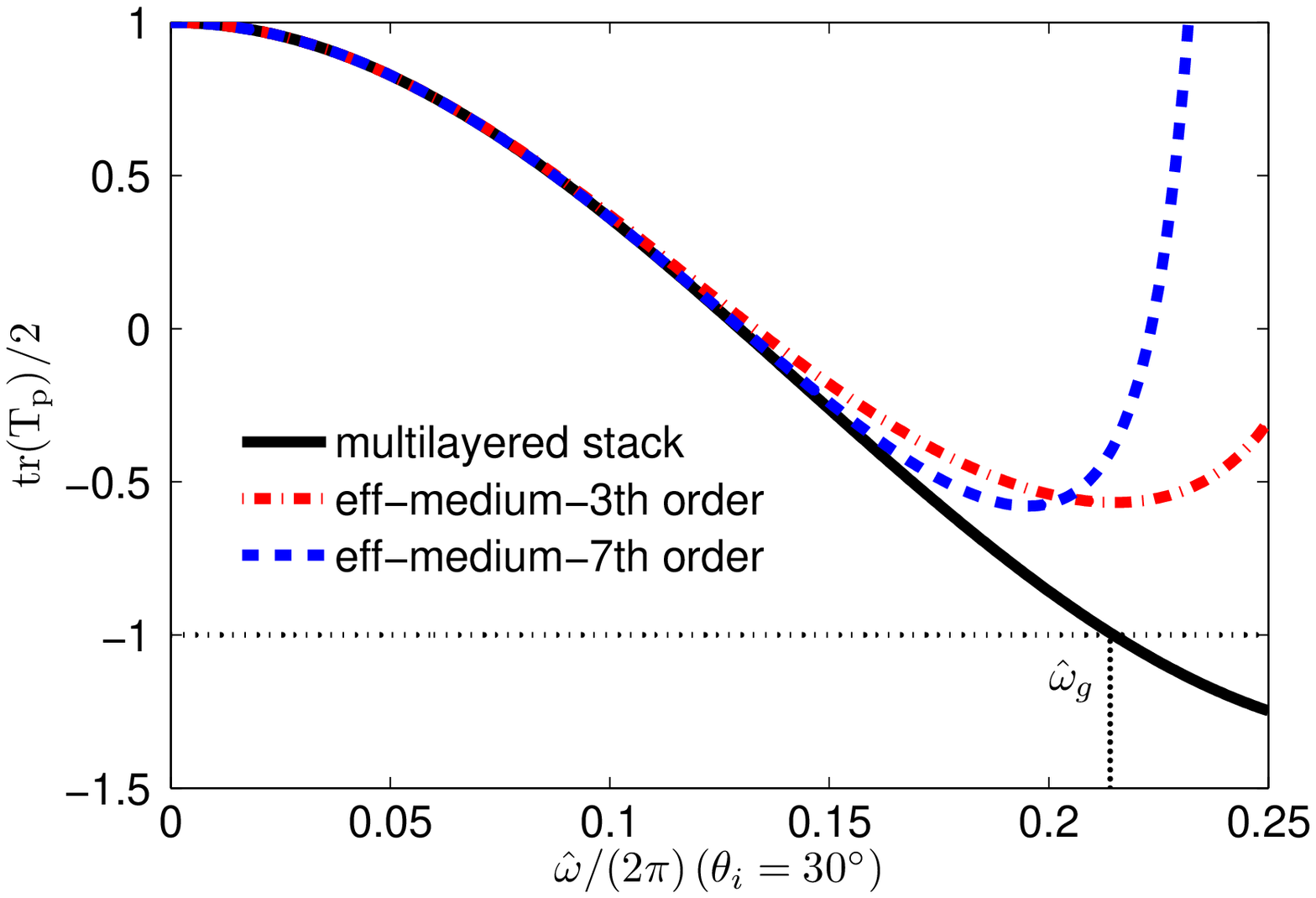}}
    \caption{Dispersion laws of the stack (solid black line) and its effective medium in 3rd order (dotted red line), and 7th order (dash blue line) approximation, under an oblique incident wave with $\theta_i=30^\circ$, (a) s-polarization with $\hat \omega_g=0.206$, (b) p-polarization with $\hat \omega_g=0.215$. }
     \label{figure4}
\end{figure}

\subsection{Transmission curve}
Apart from the dispersion law, asymptotics between the transmission curves of the multilayered stack and the effective medium is another important feature to be checked.
Figure \ref{figure5} shows a schematic diagram of the reflection and transmission for an incident wave $\text{U}^i$ on the effective medium, the thickness of which is denoted by $nh$, the upper and lower spaces of the effective medium are supposed to be vacuum with permittivity $\varepsilon_0$ and permeability $\mu_0$.
\begin{figure}[!htb]
    \centering
    \includegraphics[width=0.4\textwidth]{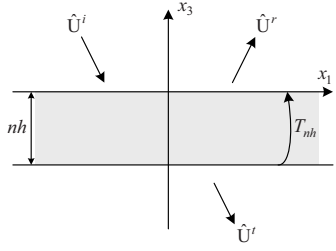}
    \caption{Schematic diagram of the reflection and transmission for an incident wave on a periodic media with thickness $nh$, the upper and lower spaces are vacuum. }
    \label{figure5}
\end{figure}

We assume the incident wave in form of Fourier decomposition in (\ref{Fourier}) is
\begin{equation}
  \widehat{\text{U}}^i= \text{U}_0 \exp[-ik_3 x_3],\quad x_3 \geq 0,
\label{int}
\end{equation}
where $k_3^2 = \omega^2 \varepsilon_0 \mu_0 - k_1^2 - k_2^2$, and $\text{U}_0$
the amplitude of the incident electromagnetic field.
The reflection and transmission waves are
\begin{equation}
\begin{array}{ll}
  \widehat{\text{U}}^r=\text{U}_0 \,  r \exp[i k_3 x_3],& x_3 \geq 0 \\[3mm]
  \widehat{\text{U}}^t=\text{U}_0 \,  t \exp[-i k_3 (x_3+nh)],& x_3 \leq -nh
\end{array}
\end{equation}
For a polarizable incident wave, the column vectors ${\widehat{\text{F}}}$ defined in (\ref{F}) are
\begin{equation}
\begin{array}{ll}
  {\rm s-polarization:} & {\widehat{\text{F}}}=[\, {\widehat {\text E}}_\perp, {\widehat {\text H}}_\parallel \,]^{\rm T} \\[2mm]
  {\rm p-polarization:} & {\widehat{\text{F}}}=[\, {\widehat {\text E}}_\parallel,{\widehat {\text H}}_\perp \,]^{\rm T}
  \label{newF}
\end{array}
\end{equation}
while for the upper and lower vacuum of the effective medium, (\ref{newF}) will be simplified as
\begin{equation}
  {\widehat{\text{F}}}=\left[\begin{array}{c}
    \widehat{\text{U}} \\
    i(\omega v_0)^{-1} \partial_{{x_3}} \widehat{\text{U}}
  \end{array}
  \right]
\end{equation}
with
\begin{equation}
\begin{array}{ll}
  {\rm s-polarization:} & \widehat{\text{U}}={\widehat {\text E}}_\perp, \, v_0=\mu_0 \,,\\[2mm]
  {\rm p-polarization:} & \widehat{\text{U}}={\widehat {\text H}}_\perp, \, v_0=\ep_0 \,.
  \label{newU}
\end{array}
\end{equation}
Assume $T_{n h}$ is the transfer matrix of the effective medium with thickness $nh$, we have
\begin{equation}
  {\widehat{\text{F}}}(0)=T_{nh} {\widehat{\text{F}}}(-nh)
\label{Ftrans}
\end{equation}
with
\begin{equation}
  T_{n h}=T^n=\left[\begin{array}{ll}
  t_{11}C_{n-1}(a)-C_{n-2}(a) & t_{12}C_{n-1}(a) \\
  t_{21}C_{n-1}(a) & t_{22}C_{n-1}(a)-C_{n-2}(a)
  \end{array}
  \right]
\label{Tnh}
\end{equation}
where $a=(t_{11}+t_{22})/2={\rm tr} (T)/2$, $T$ the transfer matrix of unit cell. And $C_n$ is the Chebyshev polynominals of the second kind \cite{Abeles50}
\begin{equation}
  C_n(a)= \dfrac{\sin[(n+1)\cos^{-1}a]}{\sqrt{1-a^2}}\,.
\end{equation}
Again, if we consider a s-polarized normal incidence, then $\widehat{\text{U}}={\widehat {\text E}}_\perp$ and $ {\widehat {\text H}}_\parallel=i(\omega \mu_0)^{-1} \partial_{x_3} {\widehat {\text E}}_\perp$, (\ref{Ftrans}) turns to be
\begin{equation}
 \left[\begin{array}{c}
    1+r \\
    \dfrac{k_3}{\omega \mu_0} (1-r)
  \end{array}
  \right]= T^{(\rm s)}_{ n h} \left[\begin{array}{c}
    t \\
    \dfrac{k_3}{\omega \mu_0} t
  \end{array}
  \right]
\end{equation}
and the transmission coefficient is
\begin{equation}
  t_{\rm s}=\dfrac{1}{(T^{(\rm s)}_{11}+T^{(\rm s)}_{22})/2+[T^{(\rm s)}_{21}+(k_3/(\omega\mu_0))^2 T^{(\rm s)}_{12}]/2}\, .
\label{ts}
\end{equation}
According to the duality between s- and p- polarization, one only needs to replace $\mu_0$ by
$\ep_0$, as well as $T_{n h}^{(\rm s)}$ by $T_{nh}^{(\rm p)}$ in (\ref{ts}) to obtain the transmission coefficient for p-polarization, e.g.
\begin{equation}
  t_{\rm p}=\dfrac{1}{(T^{(\rm p)}_{11}+T^{(\rm p)}_{22})/2+[T^{(\rm p)}_{21}+(k_3/(\omega\ep_0))^2 T^{(\rm p)}_{12}]/2}\,.
\end{equation}
Let us consider again a multilayered stack with $\ep_1=2$, $\ep_2=12$, $f_1=0.8$ and $f_2=0.2$ as an example, the thickness of the structure is supposed to be $nh$ with $n$ a constant, i.e. $n=20$, to allow a numerical calculation in Matlab. Note that the s- and p-polarized incident waves coincide under a normal incidence, i.e. $t_s=t_p=t$.
Applying $T=T_{\rm eff}$ as shown in (\ref{Tte}), and $T=T_2T_1$ with $T_m$ in (\ref{tp}) to (\ref{Tnh}) and (\ref{ts}), the transmission curves of the effective medium in 7th (dotted-dash red line) and 19th (dash blue line) order approximation are shown in figure \ref{figure6}, as well as the transmission curve of the stack (solid line).

It can be observed that the lower order approximation (dotted-dash red line) only fits well with the curve of the multilayer (solid black line) in the low frequency band; while an improved asymptotic (dash blue line) can be achieved by higher order approximation (e.g. 19th order). In agreement with the dispersion law, the asymptotics between the two transmission curves of multilayer (solid line) and effective medium in 19th order approximation (dash line) become invalid near the lower edge of the first stop band.
\begin{figure}[!htb]
    \centering
    \includegraphics[width=0.5\textwidth]{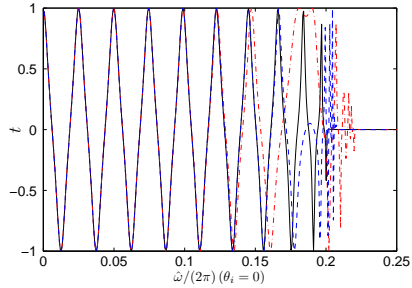}
    \caption{Transmission curves of the multilayered stack (solid line) and effective medium in 7th order (dotted-dash red line) and 19th order (dash blue line) of approximation.}
    \label{figure6}
\end{figure}

Similarly, the transmission curves of the stack and effective medium under an oblique incidence with $\theta_i=30^\circ$ in s-polarization, as well as p-polarization are shown in figure \ref{figure7}. Once more, the asymptotic approximation between the two transmission curves of the stack and the effective medium is quite good in the low frequency band, and improves with higher order approximation in effective medium as shown in the dispersion laws of figure \ref{figure4}.
\begin{figure}[!htb]
    \centering
    \subfigure[s-polarization]{\label{fig7a}
    \includegraphics[width=0.48\textwidth]{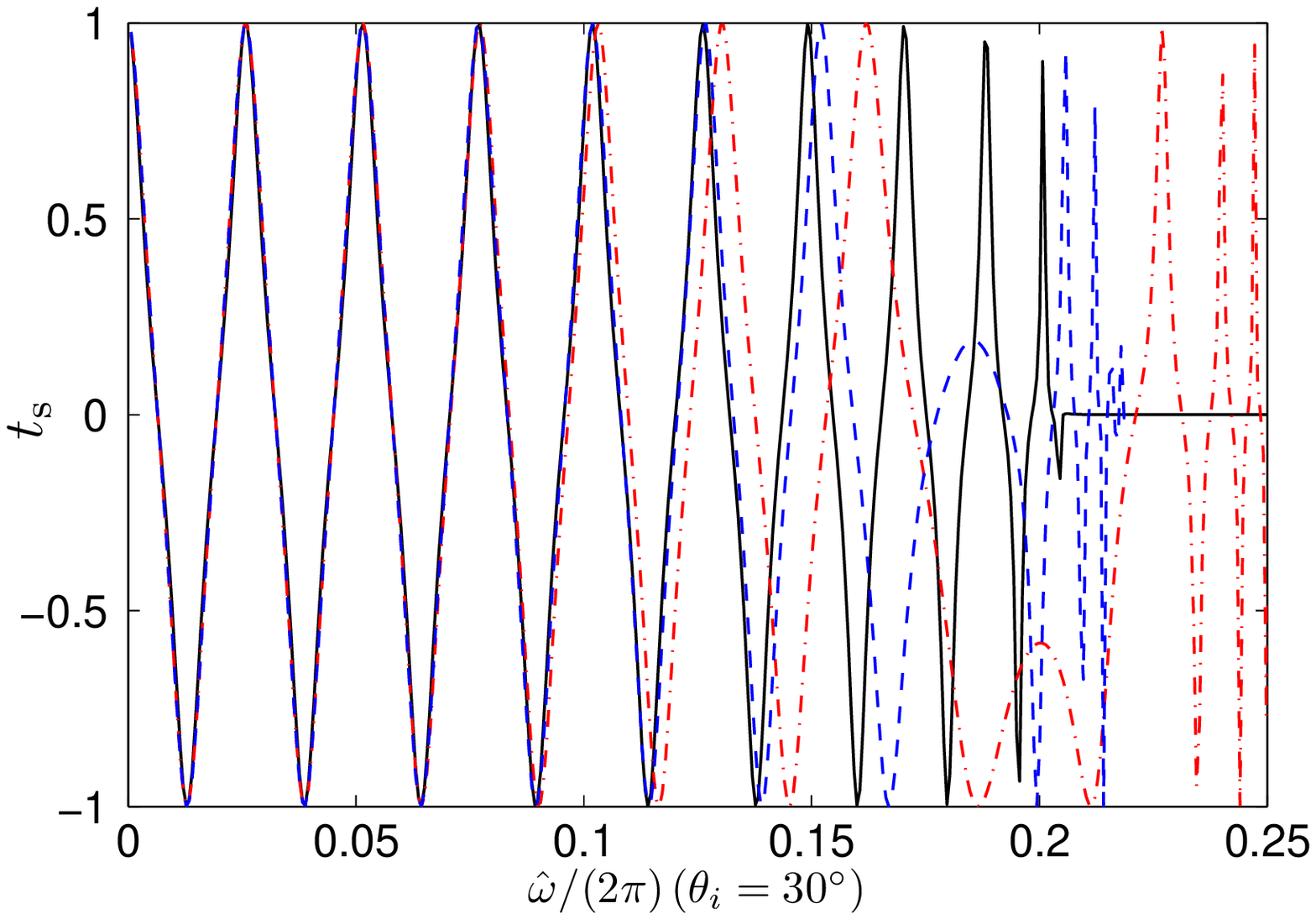}}  %scale=0.4
    \subfigure[p-polarization]{\label{fig7b}
    \includegraphics[width=0.48\textwidth]{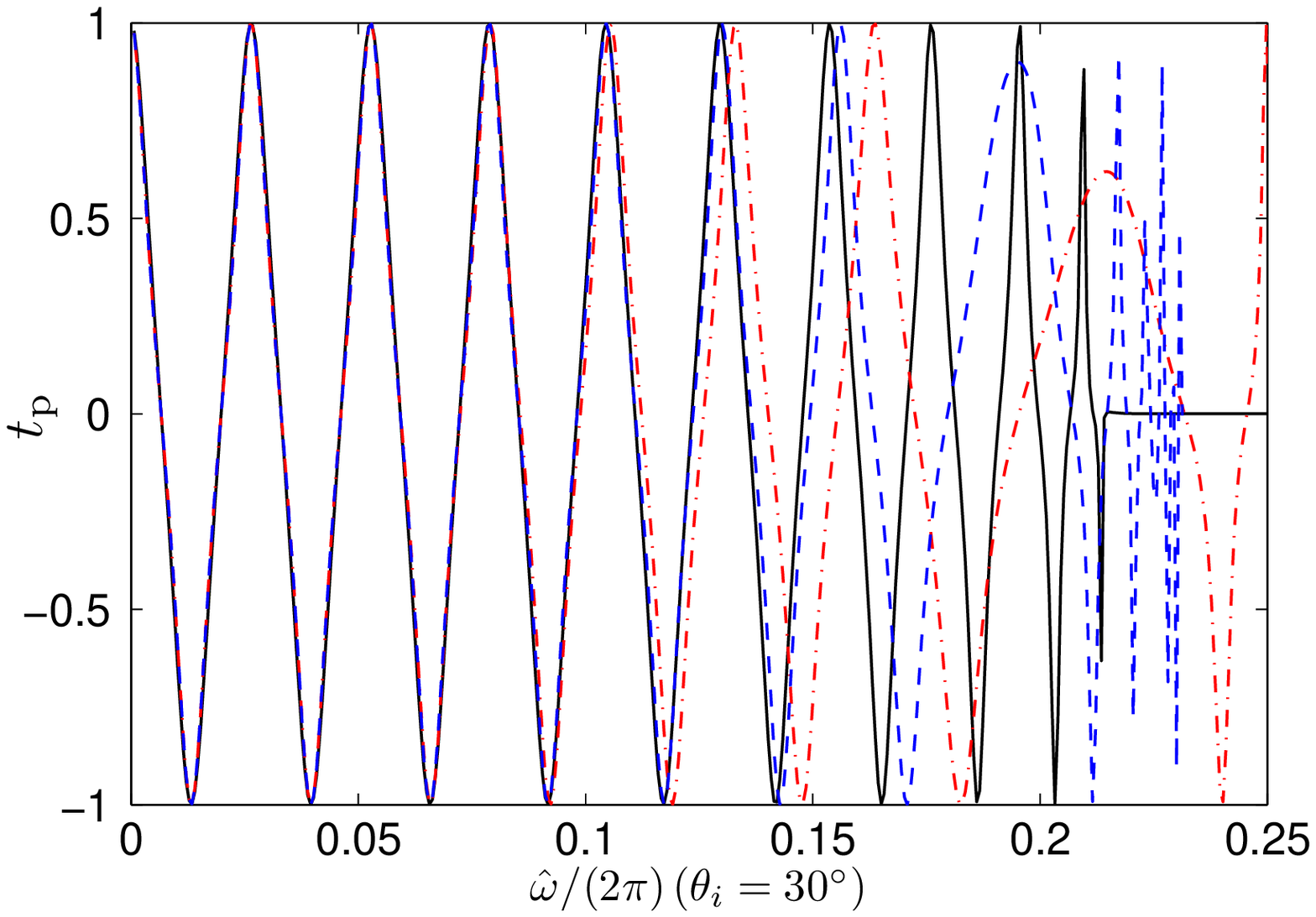}}
    \caption{ Transmission curves of the multilayered stack (solid line) and effective medium in 3rd order (dotted-dash red line) and 7th order (dash blue line) of approximation: (a) s-polarization, (b) p-polarization}
     \label{figure7}
\end{figure}

\subsection{On logarithm of transfer matrix and analyticity}
From figures \ref{figure4}-\ref{figure7}, it should be noted that asymptotics break down around the lower edge of the first stop band $\hat \omega_g$ between the dispersion laws and the transmission curves of the multilayer and its effective medium, no matter how high the order of approximation is. This invalidity is contributed to the power series expansion of $X=iM_{\rm eff}h$ in (\ref{BCH}) which diverges at $\hat \omega_g$. Indeed, we choose BCH formula to obtain the approximation for matrix $M_{\rm eff}$ furthermore for effective permittivity, permeability and bianisotropy, where we have taken $X=\log\{\exp[A]\exp[B]\}$ in (\ref{BCH}). In complex analysis, a branch of $\log(z)$ is a continuous function $L(z)$ defined on a connected open subset $G$ of the complex plane, such that $L(z)$ is a logarithm of $z$ for each $z$ in $G$ \cite{Sarason07}. An open subset $G$ is chosen as the set $\C-\R_{\leq0}$ obtained by removing the branch cut (thick solid line) along the negative real axis and the branch point (empty point) $z=0$ from the complex plane, as shown in figure \ref{figure8}(b).
\begin{figure}[!htb]
  \centering
  \includegraphics[width=0.9\textwidth]{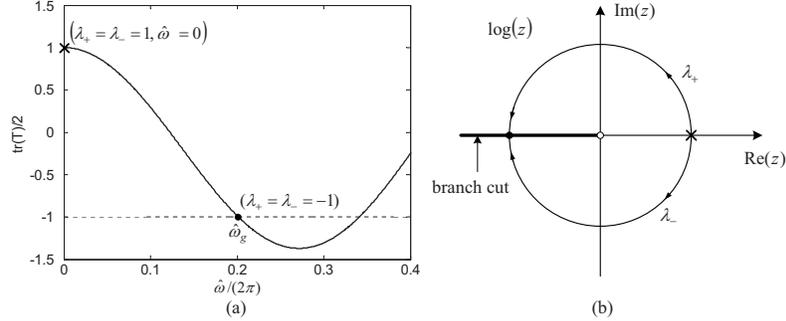}
  \caption{(a) Curve of $\rm tr(T)/2$ versus frequency for the mutlilayered stack consisting of an alternation of two dielectric layers; the eigenvalues of the transfer matrix are denoted by $\lambda_\pm$ and vary from $+1$ (cross sign) to $-1$ (filled point) in the lower pass band; while $\rm tr(T)/2=-1$ defines the edge of the first stop band; (b) Schematic diagram of a branch of $\log(z)$, a region $G$ (connected open subset of the complex plane) can be typically obtained by removing from $\C$ the interval $(-\infty,0]$. For the logarithm of $\lambda_\pm$, we choose the upper and lower half circle paths for $\lambda_\pm$ to approach $-1$ at first edge of stop band (filled point) lying on the negative real axis. }
\label{figure8}
\end{figure}

On the other hand, the transfer matrix $T=\exp[A]\exp[B]$ can be factorized by eigendecomposition as
\begin{equation}
  T= Q \Lambda Q^{-1}
\label{decom}
\end{equation}
where $Q$ is a square matrix consisting of the eigenvectors of $T$, and $\Lambda$ is a diagonal matrix whose components are the eigenvalues (denoted by $\lambda_\pm$), and
\begin{equation}
  \lambda_\pm = a \pm i \sqrt{1-a^2}
\end{equation}
with $a={\rm tr(T)}/2$ the half trace of transfer matrix $T$. Furthermore, in order to derive the effective matrix $M_{\rm eff}$ for those effective parameters, we take the logarithm of $T$
\begin{equation}
  M_{\rm eff}=\log\{T\}=Q {\rm diag}\,[\log{(\lambda_+)},\log{(\lambda_-)}] Q^{-1}\,.
\label{logM}
\end{equation}
The curve of ${\rm tr(T)}/2$ versus frequency is depicted in figure \ref{figure8}(a), and one has $\lambda_+=\lambda_-=a=+1$ at the zero frequency (denoted by cross sign), while $\lambda_+=\lambda_-=a=-1$ at the edge of the first stop band $\hat\omega_g$ (denoted by filled point). In the lower frequency band, we choose $\lambda_+$ and $\lambda_-$ (starting from $+1$) which approach $-1$ by the upper half-circle path and the lower half-circle path, respectively, where the paths lie in the open subset $G$: $\log(\lambda_\pm)$ is always analytic and unique. However, at the edge of the stop band where $\lambda_+=\lambda_-=a=-1$ on the negative real axis, the logarithm is no longer analytic when its arguments meet at the branch cut of the logarithm: This implies that an expression of $M_{\rm eff}$ as a power series of the frequency $\omega$ has its radius of convergence bounded by $\hat \omega_g$. In other words, the effective parameters lose their efficiency for the asymptotic approximation at frequencies higher than the first stop band, but they work just fine in the lower pass band. In order to achieve all frequency homogenization for a periodic structure, a new set of effective parameters (e.g. effective refractive index and surface impedance) should be introduced \cite{yan12}, where the analytic property of transfer matrix in the complex plane is ensured.

\section{Frequency power expansion of the transfer matrix}\label{sec:Tseries}
Although the function $\log\{\exp[A]\exp[B]\}$ is no longer analytic at the lower edge of the stop band, the transfer matrix $T=\exp[A]\exp[B]$ is analytic in the whole complex plane and can be approached by a power series at any frequency, a fact which will be numerically checked in this section. In mathematics, an exponential function can be approximated by a Taylor series as
\begin{equation}
  \exp[A]=\sum\limits_{n=0}^{\infty} \dfrac{A^n}{n!}=1+A+\dfrac{A^2}{2!}+\dfrac{A^3}{3!}\cdots
\end{equation}
Hence, the transfer matrix $T$ takes the form
\begin{equation}
\begin{array}{ll}
  T&=T_1T_2=\exp[iM_1h_1]\exp[iM_2h_2]=\exp[\omega N_1]\exp[\omega N_2]  \\[2mm]
  &=\left(1+\omega N_1+\dfrac{(\omega N_1)^2}{2!}+\dfrac{(\omega N_1)^3}{3!}+\cdots\right)
  \left(1+\omega N_2+\dfrac{(\omega N_2)^2}{2!}+\dfrac{(\omega N_2)^3}{3!}+\cdots\right)
\end{array}
\label{T1}
\end{equation}
with notation $\omega N_i=i M_i h_i$.
Let us expand and organize (\ref{T1}) in powers of $\omega$,
\begin{equation}
  T_{\rm eff}=T_1T_2=1+\omega (N_1+N_2)+\omega^2 (\dfrac{N_1^{2}}{2}+\dfrac{N_2^{2}}{2}+N_1 N_2)+\cdots
\label{T2}
\end{equation}
We collect the terms from $\omega^0$ to $\omega^p$ in (\ref{T2}) as an ansatz for $T_{\rm eff}$
\begin{equation}
  T_{\rm eff}\simeq T^{(0)}+\omega T^{(1)}+\omega^2 T^{(2)}+\omega^3 T^{(3)}+\cdots+\omega^p T^{(p)}\, .
\label{Texp}
\end{equation}
Obviously, with increasing $p$, the approximation in (\ref{Texp}) becomes more accurate.

Considering a normal incident wave in s-polarization, the matrices $N_m$ read as
\begin{equation}
  N_m=i h_m \left[
  \begin{array}{cc}
   0 & \mu_0 \\
   -\ep_m & 0 \\
  \end{array}
  \right], \quad m=1,\, 2
\end{equation}
so that substituting them into (\ref{Texp}) the half trace of $T_{\rm eff}$ can be expressed as
\begin{equation}
\begin{array}{ll}
{\rm tr(T_{\rm eff})}/2&= 1 - \dfrac{\hom^2}{2} \, (\ep_1 f_1 + \ep_2 f_2)
+ \dfrac{\hom^4}{24} (\ep_1 f_1 + \ep_2 f_2)^2 - \dfrac{\hom^4}{24} (\ep_1 - \ep_2)^2 f_1^2 f_2^2   \\[3mm]
&- \dfrac{\hom^6}{6!} (\varepsilon_1 f_1+\varepsilon_2f_2)^3
+ \dfrac{\hom^6}{72} f_1^2 f_2^2 (\varepsilon_1-\varepsilon_2)^2 (f_1-f_2)(\varepsilon_1 f_1-\varepsilon_2f_2) \\[3mm]
&+ \dfrac{{\hat\omega}^6}{144}  f_1^2 f_2^2 (\varepsilon_1-\varepsilon_2)^2 (\varepsilon_1 f_1+\varepsilon_2f_2)
+\dfrac{{\hat\omega}^6}{24}  f_1^3 f_2^3 (\varepsilon_1-\varepsilon_2)^2(\varepsilon_1+\varepsilon_2)   \\[3mm]
&+\dfrac{{\hat\omega}^6} {60}  f_1^2 f_2^2 (\varepsilon_1-\varepsilon_2) \left[(\varepsilon_2^2f_2-\varepsilon_1^2f_1)
+(\varepsilon_2f_1-\varepsilon_1f_2)(\varepsilon_1f_1+\varepsilon_2f_2)\right] + \cdots
\label{disp}
\end{array}
\end{equation}
where the normalized frequency $\hom=\omega h/c$. It is noted that there are only terms containing even order of $\hom$, which is due to the fact that the $T^{(2p+1)}$ in (\ref{Texp}) are all off-diagonal matrices with zero diagonal components.

The curves of $\rm tr(T)/2$ versus frequency are shown in figure \ref{figure9}(a): The thick solid line represents the half trace of transfer matrix for a multilayer consisting of an alternation of two dielectric layers with same parameters as assumed in Section \ref{sec:num}, the dash green line, dotted red line, dotted-dash blue line and the thin solid line are ${\rm tr(T_{\rm eff})}/2$ with $p=2,4,6,8$ in (\ref{Texp}) for $T_{\rm eff}$, respectively.
\begin{figure}[!htb]
    \centering
    \subfigure[]{
    \includegraphics[width=0.48\textwidth]{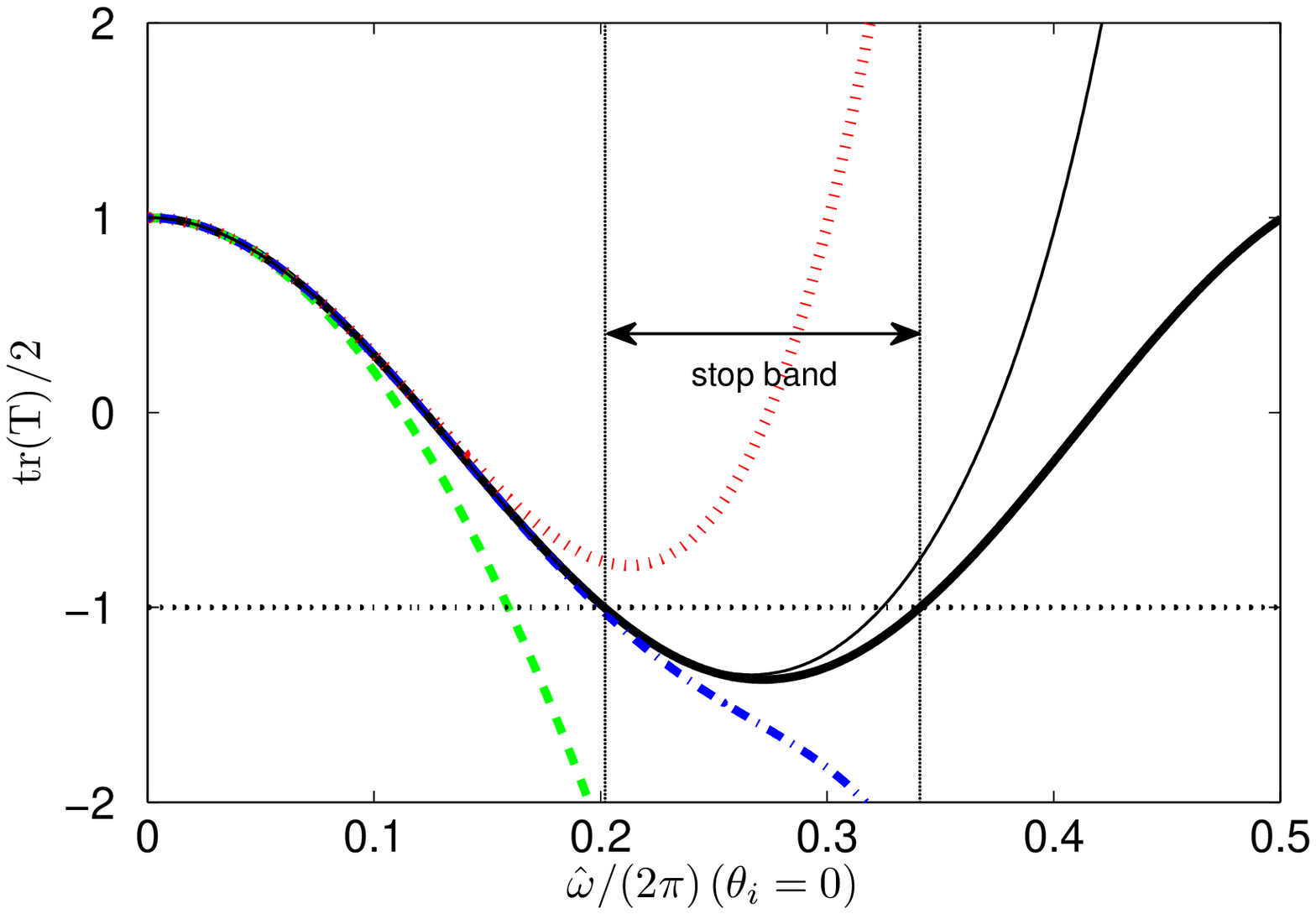}}
    \subfigure[]{
    \includegraphics[width=0.48\textwidth]{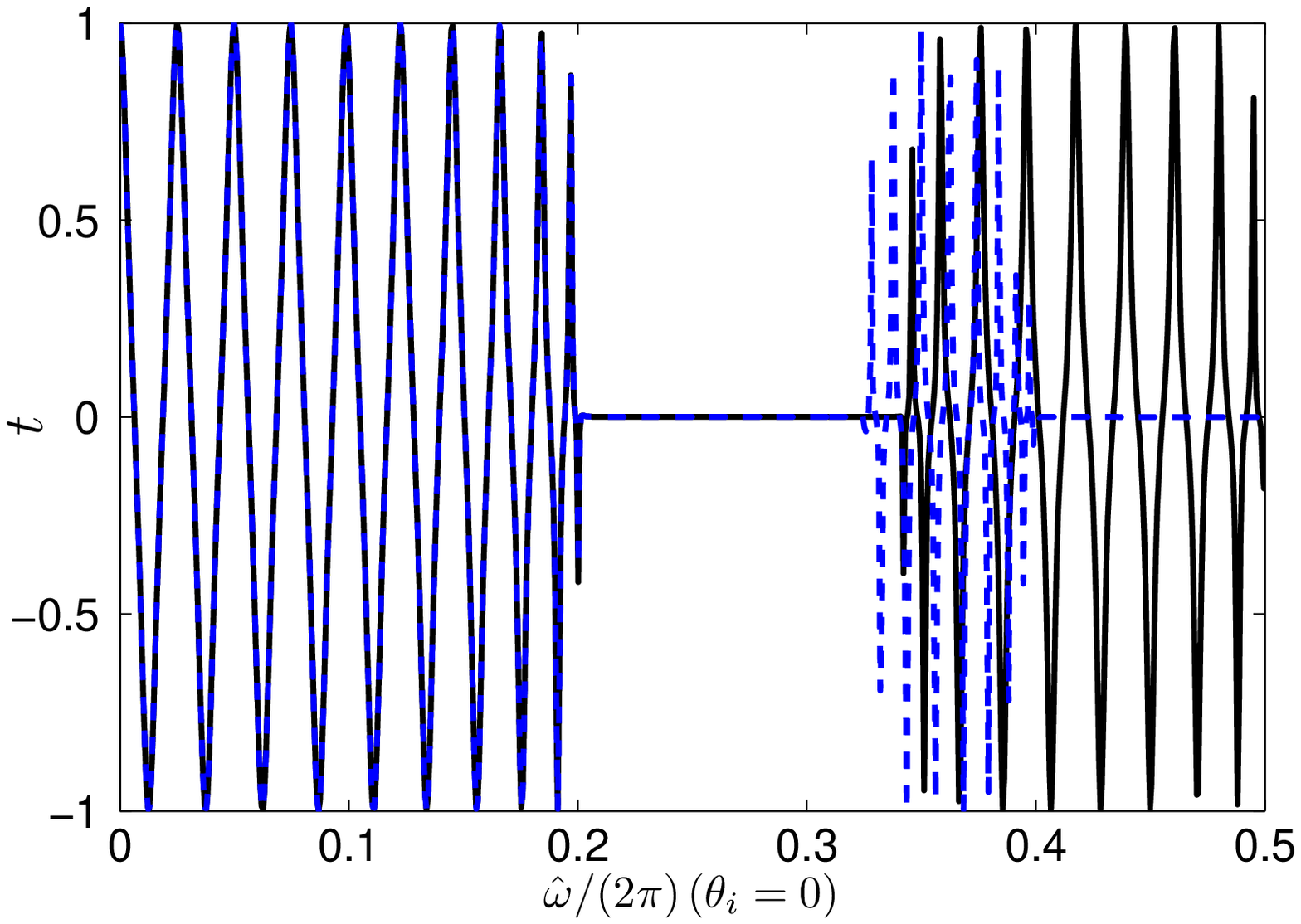}}
    \caption{(a) Curves of $\rm tr(T)/2$ versus frequency for the multilayered stack (thick solid line) and the effective medium under a normal s-polarization incidence, the dotted, dash, dotted-dash and thin solid lines correspond to the approximation for $T_{\rm eff}$ in (\ref{Texp}) by taking $p=2, 4, 6 ,8$. (b) Transmission curves of the multilayered stack (solid line) and the effective medium (dash line) with $p=8$ in (\ref{Texp}). Parameters of the dielectric layers of the stack are assumed to be the same as in Section \ref{sec:num}.}
    \label{figure9}
\end{figure}

By comparison with the curve of the multilayered stack (the thick solid line), it can be observed that the approximation for $T_{\rm eff}$ in (\ref{Texp}) with $p=2$ (dash green line) is just efficient in a range of low frequency, while $p=4$ (dotted line) gives a sharper estimate for the multilayered stack, but it totally misses the stop band. Moreover, if we push the approximation to $p=6$, the dispersion curves are nearly superimposed up to the edge of the first stop band, so well beyond the range of validity of classical homogenization \cite{Bensoussan78}. However, the approximation with $p=6$ (dotted-dash curve, which is always decreasing) breaks down at the lower edge of the stop band. In order to better approximate $T_{\rm eff}$, one needs to push the approximation to the next even number of $p$ (i.e. $p=8$, the thin solid line), which changes the curvature and gives a sharper estimate in the stop band region, although its intersection with the horizontal axis defines an approximate position for the upper edge of the stop band. This can be improved by taking larger $p$ in (\ref{Texp}). Altogether, the larger $p$ taken in (\ref{Texp}), the more accurate the approximation between the dispersion curves of the effective medium and that of the multilayered stack.

Moreover, according to the expression of the transmission coefficient in (\ref{ts}), we take $p=8$ in (\ref{Texp}) for $T_{\rm eff}$, the transmission curves for both the multilayer (solid line) and the effective medium (dash line) are depicted in figure \ref{figure9}(b). A good agreement between these two curves can be observed up to the first stop band, as predicted in figure \ref{figure9}(a). Similar calculation can be applied to an oblique incidence. This demonstrates that the transfer matrix of the effective medium can be approached as a frequency power series at any frequency.

\section{Concluding remarks} \label{sec:ccs}
We provide a rigorous high-order homogenization (HOH) algorithm for one-dimensional moderate contrast photonic crystals, where the period of the structure approaches the wavelength of optical waves. From an expression of transfer matrices in terms of exponential functions, S. Lie and BCH formulae are applied in the HOH asymptotic. Analytic expressions of the effective parameters are derived for a stack with two layers in Section \ref{sec:hoh}, where the artificial magnetism and magnetoelectric coupling effect are achieved in such a moderate contrast periodic structure. In Section \ref{sec:bch}, we explore the extension of HOH algorithm to a stack with an alternation of $m$ dielectric layers, and derive the expressions of the effective parameters for a center symmetric stack : The magnetoelectric coupling vanishes while the artificial magnetism can be achieved with non-resonant periodic structures. Furthermore, the corrector for the asymptotic approximation of a finite stack by its effective medium has been discussed in Section \ref{sec:corc}: The asymptotic error is of order $1/n^p$ with $p$ the order of the approximation. Finally, based on the expressions of the effective parameters, we numerically validate our approximation method by comparing both the dispersion law and transmission property of the stack and its effective medium in Section \ref{sec:num}. The good agreement between these curves demonstrates that the asymptotic approximation is efficient throughout the lower pass band, while at the edge of first stop band the logarithm function is no longer analytic. Finally, we investigate the approximation for the transfer matrix instead of matrix $M_{\rm eff}$ of the effective medium by a frequency power expansions, the dispersion law as well as the transmission curves of the transfer matrices for both the multialyer and the effective medium are explored in Section \ref{sec:Tseries}, and the excellent agreement confirms that the effective transfer matrix can be approached by a power series at any frequency.

\end{document}